\documentclass[11pt]{article}

\usepackage[pagebackref,colorlinks]{hyperref}
\usepackage{amsmath} 
\usepackage{amsthm} 
\usepackage{thmtools}
\usepackage{amssymb}	
\usepackage{graphicx} 
\usepackage{multicol} 
\usepackage{multirow}
\usepackage{color}
\usepackage{bm}
\usepackage[dvips,letterpaper,margin=1in,bottom=1in]{geometry}
\usepackage[capitalize,noabbrev]{cleveref}

\usepackage[utf8]{inputenc}
\usepackage[english]{babel}

\usepackage{mathtools}

\newtheorem{theorem}{Theorem}[section]

\newtheorem{lemma}[theorem]{Lemma}
\newtheorem{corollary}[theorem]{Corollary}
\newtheorem{proposition}[theorem]{Proposition}

\newtheorem{definition}{Definition}[section]

\newtheorem{problem}{Problem}[section]
\newtheorem{conjecture}{Conjecture}

\newcommand{\ketbra}[2]{\ensuremath{\ket{#1}\!\bra{#2}}}

\DeclarePairedDelimiter\rbra{\lparen}{\rparen}
\DeclarePairedDelimiter\sbra{\lbrack}{\rbrack}
\DeclarePairedDelimiter\cbra{\{}{\}}
\DeclarePairedDelimiter\abs{\lvert}{\rvert}
\DeclarePairedDelimiter\Abs{\lVert}{\rVert}

\DeclarePairedDelimiter\ket{\lvert}{\rangle}
\DeclarePairedDelimiter\bra{\langle}{\rvert}

\newcommand{\set}[2] {\left\{\, #1 \colon #2 \,\right\}}

\newcommand{\tr} {\operatorname{tr}}
\newcommand{\poly} {\operatorname{poly}}
\newcommand{\diag} {\operatorname{diag}}

\DeclarePairedDelimiter\parens{\lparen}{\rparen}
\DeclarePairedDelimiter\norm{\lVert}{\rVert}
\DeclarePairedDelimiter\braces{\lbrace}{\rbrace}

\newcommand{\calA}{\mathcal{A}}

\newcommand{\calE}{\mathcal{E}}

\newcommand{\calH}{\mathcal{H}}

\newcommand{\calP}{\mathcal{P}}

\newcommand{\calZ}{\mathcal{Z}}

\usepackage{algorithm}
\usepackage{algpseudocode}

\usepackage{tabularx}
\usepackage{booktabs}
\usepackage{threeparttable}
\usepackage{adjustbox} 

\newcommand{\footremember}[2]{%
    \footnote{#2}
    \newcounter{#1}
    \setcounter{#1}{\value{footnote}}%
}

\usepackage{tikz}
\usetikzlibrary{quantikz2}

\title{A List of Complexity Bounds for Property Testing \\ by Quantum Sample-to-Query Lifting}
\author{
    Kean Chen \footremember{1}{Kean Chen is with the Department of Computer and Information Science, University of Pennsylvania, Philadelphia, United States (e-mail: \href{mailto:keanchen.gan@gmail.com}{\nolinkurl{keanchen.gan@gmail.com}}).}
    \and Qisheng Wang \footremember{2}{Qisheng Wang is with the School of Informatics, University of Edinburgh, Edinburgh, United Kingdom (e-mail: \href{mailto:QishengWang1994@gmail.com}{\nolinkurl{QishengWang1994@gmail.com}}).}
    \and Zhicheng Zhang \footremember{3}{Zhicheng Zhang is with the Centre for Quantum Software and Information, University of Technology Sydney, Sydney, Australia (e-mail: \href{mailto:iszczhang@gmail.com}{\nolinkurl{iszczhang@gmail.com}}).}
}
\date{}

\begin{document}

\maketitle

\begin{abstract}
    Quantum sample-to-query lifting, a relation between quantum sample complexity and quantum query complexity presented in \hyperlink{cite.WZ25b}{Wang and Zhang (\textit{SIAM J.\ Comput.}\ 2025)}, was significantly strengthened by \hyperlink{cite.TWZ25}{Tang, Wright, and Zhandry (2025)} to the case of state-preparation oracles. 
    In this paper, we compile a list of quantum lower and upper bounds for property testing that are obtained by quantum sample-to-query lifting. 
    The problems of interest include testing properties of probability distributions and quantum states, such as entropy and closeness. 
    This collection contains new results, as well as new proofs of known bounds. 
    In total, we present 49 complexity bounds, where 41 are new and 18 are (near-)optimal. 
\end{abstract}

\newpage
\tableofcontents
\newpage

\section{Introduction}

Quantum sample-to-query lifting \cite{WZ25b,WZ25} reveals a novel connection between quantum query complexity and quantum sample complexity.
This discovery turns out to be useful for establishing new complexity bounds in quantum property testing \cite{MdW16}. 
The key observation is that multiple samples of a quantum state $\rho$ (sample complexity) can simulate the quantum unitary transformation $e^{-i\rho t}$ (query complexity) by density matrix exponentiation \cite{LMR14,KLL+17,GKP+25}.\footnote{The idea of simulating unitary transformations by samples of quantum states appeared even earlier in quantum communication complexity \cite{HL11}.
The case of pure states was used in quantum cryptography \cite{ARU14,JLS18,Qia24,GZ25} and the case of mixed states was used in quantum algorithms \cite{GP22,WZ24}.}

Recently, in \cite{TWZ25}, the quantum sample-to-query lifting has been thoroughly strengthened to apply to the very general case where the quantum oracle prepares a purification of the mixed quantum state. 
This input model is called \textit{purified quantum query access} \cite{GL20}, originally considered in quantum computational complexity \cite{Wat02}, and now is the standard quantum query input model for the property testing of quantum states. 

\begin{definition}[Purified quantum query access] \label{def:purified-access}
    A purified quantum query access oracle for an unknown mixed quantum state $\rho$ is described by a unitary operator $U$ acting on two subsystems $\mathsf{A}$ and $\mathsf{B}$ such that
    \[
    \rho = \tr_{\mathsf{B}}\rbra*{U_{\mathsf{AB}} \cdot \ketbra{0}{0}_{\mathsf{AB}} \cdot U_{\mathsf{AB}}^\dag}.
    \]
    In particular, a query to the oracle $U$ means a query to (controlled-)$U$ or (controlled-)$U^\dag$. 
\end{definition}

A promise problem for quantum state testing is described by a pair of disjoint sets of quantum states $\mathcal{P} = \rbra{\mathcal{P}^{\textup{yes}}, \mathcal{P}^{\textup{no}}}$.
The task is to determine with probability at least $2/3$ whether an unknown quantum state $\rho \in \mathcal{P}^{\textup{yes}}$ or $\rho \in \mathcal{P}^{\textup{no}}$, given the promise that it is in either case. 
The sample complexity of $\mathcal{P}$, denoted as $\mathsf{S}\rbra{\mathcal{P}}$, is the minimum number of samples of $\rho$ for solving $\mathcal{P}$. 
The query complexity of $\mathcal{P}$, denoted as $\mathsf{Q}\rbra{\mathcal{P}}$, is the minimum number of queries to the purified quantum query access oracle for solving $\mathcal{P}$. 

Using the above notations, \cite[Theorem 1.5]{TWZ25} can be rephrased as the following theorem. 

\begin{theorem}[Quantum sample-to-query lifting for purified quantum query access \cite{TWZ25}] \label{thm:main}
    Let $\mathcal{P}$ be a promise problem for quantum state testing. 
    Then, 
    \[
        \mathsf{Q}\rbra{\mathcal{P}} = \Omega\rbra*{\sqrt{\mathsf{S}\rbra{\mathcal{P}}}}. 
    \]
\end{theorem}

Prior to \cref{thm:main}, similar results about the sample complexity in terms of mixed states and their purifications have appeared in the quantum property testing literature \cite[Theorem 5.2]{SW22} and \cite{CWZ24}. 
The relation between quantum sample complexity and quantum query complexity revealed in \cref{thm:main} is somewhat surprising, since universal controllization \cite{AFCB14,GST24} and purification \cite{LDCL25} are impossible and unitary inversion requires exponential resources \cite{FK18,QDS+19a,QDS+19b,CSM23,GST24,CML+24,OYM24,YKS+24,SHH25,CYZ25}.

As a corollary, \cref{thm:main} reproduces (and improves) the quantum sample-to-query lifting theorem in \cite[Theorem 1.1]{WZ25b} for the diamond quantum query complexity $\mathsf{Q}_{\diamond}^{\alpha}\rbra{\mathcal{P}}$, which states that $\mathsf{Q}_{\diamond}^{\alpha}\rbra{\mathcal{P}} \geq \widetilde{\Omega}\rbra{\sqrt{\mathsf{S}\rbra{\mathcal{P}}}}$ for any constant $\alpha > 1$.\footnote{$\widetilde{O}\rbra{\cdot}$ and $\widetilde{\Omega}\rbra{\cdot}$ suppress the polylogarithmic factors.}
Here, $\mathsf{Q}_{\diamond}^{\alpha}\rbra{\mathcal{P}}$ is the minimum number of queries to a unitary oracle $U$ with upper left corner block $\rho/\alpha$, where $\rho$ is the quantum state to be tested and $\alpha \geq 1$ is a parameter. 
In the language of block-encoding \cite{GSLW19}, the unitary oracle $U$ is a block-encoding of $\rho/\alpha$.
\cref{thm:main} gives the following stronger statement, which (i) allows $\alpha$ to be $1$, and (ii) removes the polylogarithmic factors. 

\begin{corollary}[Improved quantum sample-to-query lifting for block-encoded access] \label{corollary:lifting-block-encoding}
    Let $\mathcal{P}$ be a promise problem for quantum state testing. Then, for any $\alpha \geq 1$, 
    \[
        \mathsf{Q}_{\diamond}^{\alpha}\rbra{\mathcal{P}} = \Omega\rbra*{\sqrt{\mathsf{S}\rbra{\mathcal{P}}}}. 
    \]
\end{corollary}

\begin{proof}
    Let $\mathcal{O}$ be a unitary operator that prepares a purification of a quantum state $\rho$. 
    Then, we can implement a unitary operator $U$ with upper left corner block $\rho$ using $O\rbra{1}$ queries to $\mathcal{O}$ by \cite[Lemma 7]{LC19} and \cite[Lemma 25]{GSLW19}. 
    Therefore, \cref{thm:main} immediately yields the proof for the case of $\alpha = 1$. 
    The case of $\alpha > 1$ can be obtained by the fact \cite[Theorem B.1]{WZ25b} that $\mathsf{Q}_\diamond^{\alpha}\rbra{\mathcal{P}} \geq \mathsf{Q}_\diamond^{1}\rbra{\mathcal{P}}$ for any $\alpha > 1$. 
\end{proof}

For the special case where the quantum state $\rho$ is pure, we can obtain a quantum sample-to-query lifting in terms of the reflection operator. 
Let $\mathsf{Q}_\diamond^{\mathsf{refl}}\rbra{\mathcal{P}}$ be the diamond quantum query complexity of a pure-state quantum state testing problem $\mathcal{P}$, which is the minimum number of queries to the reflection operator $R_{\psi} = I - 2\ketbra{\psi}{\psi}$ for the input pure state $\ket{\psi}$ to be tested. 

\begin{corollary} [Improved quantum sample-to-query lifting for reflection access] \label{corollary:lifting-reflection}
    Let $\mathcal{P}$ be a promise problem for pure-state quantum state testing. 
    Then, 
    \[
    \mathsf{Q}_\diamond^{\mathsf{refl}}\rbra{\mathcal{P}} = \Omega\rbra*{\sqrt{\mathsf{S}\rbra{\mathcal{P}}}}. 
    \]
\end{corollary}
\begin{proof}
    This is immediate by \cref{corollary:lifting-block-encoding} and noting that $\mathsf{Q}_\diamond^{\mathsf{refl}}\rbra{\mathcal{P}} = \Theta\rbra{\mathsf{Q}_{\diamond}^{1}\rbra{\mathcal{P}}}$ for any pure-state quantum state testing problem $\mathcal{P}$ (cf.\ \cite[Lemma 5.6]{CWZ24}). 
\end{proof}

\cref{thm:main} can be restated below as a samplizer of quantum query algorithms with purified quantum query access. 

\begin{corollary}[Quantum samplizer, \cref{thm:main} restated] \label{thm:samplizer}
    Let $\mathcal{P}$ be a promise problem for quantum state testing. 
    Then, 
    \[
        \mathsf{S}\rbra{\mathcal{P}} = O\rbra*{\rbra*{\mathsf{Q}\rbra{\mathcal{P}}}^2}. 
    \]
\end{corollary}

\cref{thm:samplizer} can be used to derive quantum sample upper bounds from quantum query upper bounds. 
\cref{thm:samplizer} generalizes the samplizer given in \cite{WZ25}. 
In comparison, the samplizer in \cref{thm:samplizer} applies to the purified quantum query access model, whereas the samplizer in \cite{WZ25} applies to a weaker input model where the unitary oracle is a block-encoding of $\rho$. 
The time complexity (i.e., the number of elementary quantum gates) of the quantum algorithm obtained by \cref{thm:samplizer} is $\widetilde{O}\rbra{n^4}$ due to the use of the Schur transform \cite{Har05,BCH06,BCH07,KS18,Kro19,Ngu23,GBO23,BFG+25} in the approach of \cite{TWZ25}, where $n$ is the sample complexity and the time complexity is obtained by the current best Schur transform \cite{BFG+25}. 
In comparison, the time complexity of the quantum algorithm obtained by the samplizer in \cite{WZ25} is $\widetilde{O}\rbra{n}$. 

\vspace{4pt}

It turns out that \cref{thm:main} and its variants (\cref{corollary:lifting-block-encoding,corollary:lifting-reflection,thm:samplizer}) are able to derive new quantum complexity lower and upper bounds for a number of problems ranging from the classical world to the quantum world. 
In this paper, we compile a list of quantum lower and upper bounds that can be proved by quantum sample-to-query lifting. 

\subsection{Spotlight}

In summary, this paper includes 48 quantum complexity bounds obtained by quantum sample-to-query lifting, where 40 are new and 19 are (near-)optimal. 
The statistics is given in \cref{tab:count}. 

\begin{table}[!htp]
\centering
\caption{Statistics of quantum complexity bounds by quantum sample-to-query lifting.}\label{tab:count}
\adjustbox{max width=\textwidth}{
\begin{tabular}{ccccc}
\toprule
Collection  & \#Bounds & \#New & \#Optimal & \#New\&Optimal   \\
\midrule
Testing Classical Properties (\cref{tab:classical}) & 12 & 8 & 10 & 7 \\ \midrule
Testing Quantum Properties (\cref{tab:quantum}) & 19 & 18 & 2 & 1 \\ \midrule
Other Quantum Tasks (\cref{tab:other}) & 8 & 5 & 5 & 2 \\ \midrule
Sample Upper Bounds (\cref{tab:sample}) & 10 & 10 & 1 & 1 \\ \midrule \midrule
Total & 49 & 41 & 18 & 11 \\
\bottomrule
\end{tabular}
}
\end{table}

\subsubsection{Quantum query lower bounds}

We present a list of quantum query lower bounds obtained by quantum sample-to-query lifting. 
In particular, some lower bounds are tight. 
This shows that the quantum sample-to-query lifting method brings new methodologies into proving quantum query lower bounds, in addition to the quantum polynomial method \cite{BBC+01}, the quantum adversary method \cite{Amb02}, and the quantum compressed oracle method \cite{Zha19}. 

\paragraph{Testing classical properties.}

In \cref{tab:classical}, $12$ quantum query lower bounds for testing classical properties are listed, including hypothesis testing, uniformity testing, closeness testing, and entropy estimation of probability distributions, where $8$ are new and $10$ are near-optimal. 

In particular, we present 
\begin{enumerate}
    \item A new proof (see \cref{thm:Bel19-reproduced}) of the quantum query lower bound for hypothesis testing, which was originally proved in \cite{Bel19} by the quantum adversary method. 
    \item Improved quantum query lower bounds for the estimation of total variation distance, Shannon entropy, and support size (see \cref{thm:tvd-estimation,thm:lb-shannon,thm:support-size}), where the previous lower bounds are due to \cite{BKT20} by the quantum polynomial method. 
    \item The \textit{first} almost matching quantum query lower bound for R\'enyi entropy estimation (see \cref{thm:lt1-renyi}), showing that the quantum algorithm given in \cite[Theorem 1]{WZL24} is near-optimal in the dimension $d$ for $0 < \alpha < 1$.
    \item The \textit{first} non-trivial quantum query lower bound for $2$-wise uniformity testing (see \cref{thm:k-wise}), showing that the quantum algorithm given in \cite[Theorem 5]{LWL24} is optimal in the precision parameter $\varepsilon$. 
\end{enumerate}

\begin{table}[t]
\centering
\caption{Quantum query lower bounds for testing classical properties.}\label{tab:classical}
\adjustbox{max width=\textwidth}{
\begin{tabular}{ccccc}
\toprule
Problem  & Previous & This Paper & New?       & Optimal?     \\
\midrule
Hypothesis Testing       & \begin{tabular}{c}
     $\Omega\rbra{1/d_{\textup{H}}\rbra{P, Q}}$  \\
     \cite{Bel19} 
\end{tabular}         & \begin{tabular}{c}
     $\Omega\rbra{1/d_{\textup{H}}\rbra{P, Q}}$  \\
     \cref{thm:Bel19-reproduced} 
\end{tabular}           & Reproved        & \textbf{Optimal}      \\ \midrule
$\ell_1$ Closeness Testing       & \begin{tabular}{c}
     $\Omega\rbra{d^{1/3}+\frac{1}{\varepsilon}}$  \\
     \cite{CFMdW10,LWL24} 
\end{tabular}         & \begin{tabular}{c}
     $\Omega\rbra{\frac{d^{1/3}}{\varepsilon^{2/3}}+\frac{d^{1/4}}{\varepsilon}}$  \\
     \cref{thm:ell1-improved} 
\end{tabular}           & \textbf{New}        & \begin{tabular}{c} \textbf{Optimal} \\ in $d, \varepsilon$ \end{tabular}     \\ \midrule
Uniformity Testing       & \begin{tabular}{c}
     $\Omega\rbra{d^{1/3}}$  \\
     \cite{CFMdW10} 
\end{tabular}         & \begin{tabular}{c}
     $\Omega\rbra{d^{1/3}+\frac{d^{1/4}}{\varepsilon}}$  \\
     \cref{thm:lb-uniformity} 
\end{tabular}           & \textbf{New}        & \begin{tabular}{c} \textbf{Optimal} \\ in $d, \varepsilon$ \end{tabular}      \\ \midrule
$\ell_2$ Closeness Testing       & \begin{tabular}{c}
     $\Omega\rbra{\frac{1}{\varepsilon}}$  \\
     \cite{LWL24} 
\end{tabular}         & \begin{tabular}{c}
     $\Omega\rbra{\frac{1}{\varepsilon}}$  \\
     \cref{thm:ell2-rep} 
\end{tabular}           & Reproved        & \textbf{Optimal}      \\ \midrule
\begin{tabular}{c} Total Variation \\ Distance Estimation \end{tabular}        & \begin{tabular}{c}
     $\widetilde{\Omega}\rbra{\sqrt{d}}$  \\
     \cite{BKT20} 
\end{tabular}         & \begin{tabular}{c}
     $\Omega\rbra{\frac{\sqrt{d}}{\varepsilon\sqrt{\log\rbra{d/\varepsilon}}}}$  \\
     \cref{thm:tvd-estimation} 
\end{tabular}           & \textbf{New}        & \begin{tabular}{c} Near-\textbf{Optimal} \\ in $d$ \end{tabular}     \\ \midrule
\begin{tabular}{c} Shannon Entropy \\ Estimation \end{tabular}       & \begin{tabular}{c}
     $\widetilde{\Omega}\rbra{\sqrt{d}}$  \\
     \cite{BKT20} 
\end{tabular}         & \begin{tabular}{c}
     $\Omega\rbra{\frac{\sqrt{d}}{\sqrt{\varepsilon \log\rbra{d}}}+\frac{\log\rbra{d}}{\varepsilon}}$  \\
     \cref{thm:lb-shannon} 
\end{tabular}           & \textbf{New}        & \begin{tabular}{c} Near-\textbf{Optimal} \\ in $d, \varepsilon$ \end{tabular}     \\ \midrule
\begin{tabular}{c} $\alpha$-R\'enyi Entropy \\ Estimation ($0 < \alpha < 1$) \end{tabular}       & \begin{tabular}{c}
     $\Omega\rbra{\frac{d^{\frac{1}{3}}}{\varepsilon^{\frac{1}{6}}} + \frac{d^{\frac{1}{2\alpha}-\frac{1}{2}}}{\varepsilon^{\frac{1}{2\alpha}}}}$  \\
     \cite{LW19,WZL24} 
\end{tabular}         & \begin{tabular}{c}
     $\Omega\rbra{d^{\frac{1}{2\alpha}-o\rbra{1}}+\frac{d^{\frac{1}{2\alpha}-\frac{1}{2}}}{\varepsilon^{\frac{1}{2\alpha}}}}$  \\
     \cref{thm:lt1-renyi} 
\end{tabular}           & \textbf{New}        & \begin{tabular}{c} Near-\textbf{Optimal} \\ in $d$ \end{tabular}     \\ 
\begin{tabular}{c} $\alpha$-R\'enyi Entropy \\ Estimation ($\alpha > 1$) \end{tabular}       & \begin{tabular}{c}
     $\Omega\rbra{\frac{d^{\frac{1}{3}}}{\varepsilon^{\frac{1}{6}}}+\frac{d^{\frac{1}{2}-\frac{1}{2\alpha}}}{\varepsilon}}$  \\
     \cite{LW19} 
\end{tabular}         & \begin{tabular}{c}
     $\Omega\rbra{d^{\frac{1}{2}-o\rbra{1}}+\frac{d^{\frac{1}{2}-\frac{1}{2\alpha}}}{\varepsilon}}$  \\
     \cref{thm:renyi-gt-1} 
\end{tabular}           & \textbf{New}        & Not Yet     \\
\begin{tabular}{c} $\alpha$-R\'enyi Entropy \\ Estimation (Integer $\alpha$) \end{tabular}       & \begin{tabular}{c}
     $\Omega\rbra{\frac{d^{\frac{1}{2}-\frac{1}{2\alpha}}}{\varepsilon}}$  \\
     \cite{LW19} 
\end{tabular}         & \begin{tabular}{c}
     $\Omega\rbra{\frac{d^{\frac{1}{2}-\frac{1}{2\alpha}}}{\varepsilon}}$  \\
     \cref{thm:integer-renyi} 
\end{tabular}           & Reproved        & Not Yet  \\ \midrule
\begin{tabular}{c} $\alpha$-Tsallis Entropy \\ Estimation (Integer $\alpha$) \end{tabular}       & \begin{tabular}{c}
     $\Omega\rbra{\frac{1}{\sqrt{\alpha}\varepsilon}}$  \\
     \cite{Wan25} 
\end{tabular}         & \begin{tabular}{c}
     $\Omega\rbra{\frac{1}{\sqrt{\alpha}\varepsilon}}$  \\
     \cref{thm:integer-tsallis} 
\end{tabular}           & Reproved        & \begin{tabular}{c} \textbf{Optimal} in $\alpha$, \\ Near-\textbf{Optimal} \\ in $\varepsilon$ \end{tabular} \\ \midrule
\begin{tabular}{c} Support Size \\ Estimation \end{tabular}       & \begin{tabular}{c}
     $\widetilde{\Omega}\rbra{\sqrt{d}}$  \\
     \cite{BKT20} 
\end{tabular}         & \begin{tabular}{c}
     $\Omega\rbra{\frac{\sqrt{d}}{\sqrt{\log\rbra{d}}}\log\rbra{\frac{1}{\varepsilon}}}$  \\
     \cref{thm:support-size} 
\end{tabular}           & \textbf{New}        & \begin{tabular}{c} Near-\textbf{Optimal} \\ in $d$ \end{tabular}\\ \midrule
\begin{tabular}{c} $2$-Wise \\ Uniformity Testing \end{tabular}       & N/A         & \begin{tabular}{c}
     $\Omega\rbra{\frac{\sqrt{n}}{\varepsilon}}$  \\
     \cref{thm:k-wise} 
\end{tabular}           & \textbf{New}        & \textbf{Optimal} in $\varepsilon$  \\ 
\bottomrule
\end{tabular}
}
\end{table}

\paragraph{Testing quantum properties.}
In \cref{tab:quantum}, $19$ quantum query lower bounds for testing quantum properties are listed, including mixedness testing, rank testing, entropy estimation, and closeness estimation of quantum states, where $18$ are new and $2$ are near-optimal. 

In particular, we present improved quantum query lower bounds for a series of property testing problems, including mixedness testing (\cref{thm-10150001}), von Neumann entropy estimation (\cref{thm:von-neumann-entropy}), trace distance estimation (\cref{thm:lb-td}), and fidelity estimation (\cref{thm:lb-fidelity}), which are based on the quantum sample lower bound for mixedness testing and significantly improve the previous known lower bounds implied by or due to \cite{CFMdW10,BKT20,Wan24,UNWT25}. 

\begin{table}[!htp]
\centering
\caption{Quantum query lower bounds for testing quantum properties.} \label{tab:quantum}
\adjustbox{max width=\textwidth}{
\begin{tabular}{ccccc}
\toprule
Problem  & Previous & This Paper & New?       & Optimal?     \\
\midrule
Mixedness Testing       & \begin{tabular}{c}
     $\Omega\rbra{d^{1/3}}$  \\
     \cite{CFMdW10} 
\end{tabular}         & \begin{tabular}{c}
     $\Omega(\frac{\sqrt{d}}{\varepsilon})$  \\
     \cref{thm-10150001} 
\end{tabular}           & \textbf{New}        & Not Yet      \\ 
State Certification       & N/A         & \begin{tabular}{c}
     $\Omega(\frac{\sqrt{d}}{\varepsilon})$  \\
     \cref{corollary:state-cert} 
\end{tabular}           & \textbf{New}        & Not Yet      \\ 
\begin{tabular}{c} Mixedness Testing \\ for Block-Encoding \end{tabular}      & \begin{tabular}{c}
     $\widetilde{\Omega}(\frac{\sqrt{d}}{\varepsilon})$  \\
     \cite{WZ25b}
\end{tabular}         & \begin{tabular}{c}
     $\Omega(\frac{\sqrt{d}}{\varepsilon})$  \\
     \cref{corollary:mixedness-be} 
\end{tabular}           & \textbf{New}        & Not Yet      \\ \midrule
Rank Testing       & N/A         & \begin{tabular}{c}
     $\Omega(\frac{\sqrt{r}}{\sqrt{\varepsilon}})$  \\
     \cref{thm:rank-testing} 
\end{tabular}           & \textbf{New}        & Not Yet      \\ 
\begin{tabular}{c} Rank Testing \\ for Block-Encoding \end{tabular}       & \begin{tabular}{c}
     $\widetilde{\Omega}(\frac{\sqrt{r}}{\sqrt{\varepsilon}})$  \\
     \cite{WZ25b}
\end{tabular}         & \begin{tabular}{c}
     $\Omega(\frac{\sqrt{r}}{\sqrt{\varepsilon}})$  \\
     \cref{corollary:rank-be} 
\end{tabular}           & \textbf{New}        & Not Yet      \\ \midrule
Uniformity Distinguishing       & N/A         & \begin{tabular}{c}
     $\Omega^*(\frac{r}{\sqrt{\Delta}})$  \\
     \cref{thm:uniformity-distinguishing} 
\end{tabular}           & \textbf{New}        & Not Yet      \\ 
\begin{tabular}{c} Uniformity Distinguishing \\ for Block-Encoding \end{tabular}       & \begin{tabular}{c}
     $\widetilde{\Omega}(\sqrt{r})$  \\
     \cite{WZ25b}
\end{tabular}         & \begin{tabular}{c}
     ${\Omega}(\sqrt{r})$  \\
     \cref{corollary:uniform-dis-be} 
\end{tabular}           & \textbf{New}        & Not Yet      \\ \midrule
\begin{tabular}{c}
     Maximal Entanglement \\
     Testing
\end{tabular}       & N/A         & \begin{tabular}{c}
     $\Omega(\frac{\sqrt{d}}{\varepsilon})$  \\
     \cref{thm:maximal-entanglement} 
\end{tabular}           & \textbf{New}        & Not Yet      \\ \midrule
\begin{tabular}{c}
     Uniform Schmidt \\
     Testing
\end{tabular}       & N/A         & \begin{tabular}{c}
     $\Omega(\frac{\sqrt{r}}{\varepsilon})$  \\
     \cref{thm:schmidt-rank} 
\end{tabular}           & \textbf{New}        & Not Yet      \\ \midrule
\begin{tabular}{c}
     Uniform Schmidt \\
     Coefficient Testing
\end{tabular}       & N/A         & \begin{tabular}{c}
     $\Omega^*(\frac{r}{\sqrt{\Delta}})$  \\
     \cref{thm:uniform-schmidt} 
\end{tabular}           & \textbf{New}        & Not Yet      \\ \midrule
\begin{tabular}{c}
     Von Neumann \\
     Entropy Estimation
\end{tabular}       & \begin{tabular}{c}
     $\widetilde{\Omega}\rbra{\sqrt{d}}$  \\
     \cite{BKT20} 
\end{tabular}         & \begin{tabular}{c}
     $\Omega(\frac{\sqrt{d}}{\sqrt{\varepsilon}}+\frac{\log\rbra{d}}{\varepsilon})$  \\
     \cref{thm:von-neumann-entropy} 
\end{tabular}           & \textbf{New}        & Not Yet      \\ \midrule
\begin{tabular}{c} $\alpha$-R\'enyi Entropy \\ Estimation ($0 < \alpha < 1$) \end{tabular}       & \begin{tabular}{c}
     $\Omega\rbra{\frac{d^{\frac{1}{3}}}{\varepsilon^{\frac{1}{6}}} + \frac{d^{\frac{1}{2\alpha}-\frac{1}{2}}}{\varepsilon^{\frac{1}{2\alpha}}}}$  \\
     \cite{LW19,WZL24} 
\end{tabular}         & \begin{tabular}{c}
     $\Omega\rbra{d^{\frac{1}{2\alpha}-o\rbra{1}}+\frac{d^{\frac{1}{2\alpha}-\frac{1}{2}}}{\varepsilon^{\frac{1}{2\alpha}}}+\frac{\sqrt{d}}{\sqrt{\varepsilon}}}$  \\
     \cref{thm:qRenyi} 
\end{tabular}           & \textbf{New}        & Not Yet    \\
\begin{tabular}{c} $\alpha$-R\'enyi Entropy \\ Estimation ($\alpha > 1$) \end{tabular}       & \begin{tabular}{c}
     $\Omega\rbra{\frac{d^{\frac{1}{3}}}{\varepsilon^{\frac{1}{6}}}+\frac{d^{\frac{1}{2}-\frac{1}{2\alpha}}}{\varepsilon}}$  \\
     \cite{LW19} 
\end{tabular}         & \begin{tabular}{c}
     $\Omega\rbra{\frac{d^{\frac{1}{2}-\frac{1}{2\alpha}}}{\varepsilon}+\frac{\sqrt{d}}{\sqrt{\varepsilon}}}$  \\
     \cref{thm:qRenyi} 
\end{tabular}           & \textbf{New}        & Not Yet     \\
\begin{tabular}{c} $\alpha$-R\'enyi Entropy \\ Estimation (Integer $\alpha$) \end{tabular}       & \begin{tabular}{c}
     $\Omega\rbra{\frac{d^{\frac{1}{3}}}{\varepsilon^{\frac{1}{6}}}+\frac{d^{\frac{1}{2}-\frac{1}{2\alpha}}}{\varepsilon}}$  \\
     \cite{LW19} 
\end{tabular}         & \begin{tabular}{c}
     $\Omega(\frac{d^{1-\frac{1}{\alpha}}}{\varepsilon^{\frac{1}{\alpha}}}+\frac{d^{\frac{1}{2}-\frac{1}{2\alpha}}}{\varepsilon})$  \\
     \cref{thm:qRenyi} 
\end{tabular}           & \textbf{New}        & Not Yet     \\ \midrule
\begin{tabular}{c}
     $\alpha$-Tsallis Entropy \\
     Estimation ($1 < \alpha < \frac{3}{2}$)
\end{tabular}       & \begin{tabular}{c}
     $\Omega\rbra{\frac{1}{\sqrt{\varepsilon}}}$  \\
     \cite{LW25} 
\end{tabular}         & \begin{tabular}{c}
     $\Omega\rbra{\frac{1}{\varepsilon^{\frac{1}{2\rbra{\alpha-1}}}}}$  \\
     \cref{thm:qTsallis} 
\end{tabular}           & \textbf{New}        & Not Yet      \\ 
\begin{tabular}{c}
     $\alpha$-Tsallis Entropy \\
     Estimation ($\alpha \geq \frac{3}{2}$)
\end{tabular}       & \begin{tabular}{c}
     $\Omega\rbra{\frac{1}{\sqrt{\varepsilon}}}$  \\
     \cite{LW25} 
\end{tabular}         & \begin{tabular}{c}
     $\Omega\rbra{\frac{1}{\varepsilon}}$  \\
     \cref{thm:qTsallis} 
\end{tabular}           & \textbf{New}        & \begin{tabular}{c} \textbf{Optimal} \\ for integer $\alpha \geq 2$ \end{tabular}     \\ \midrule
\begin{tabular}{c}
     Trace Distance \\
     Estimation
\end{tabular}       & \begin{tabular}{c}
     $\widetilde{\Omega}\rbra{\sqrt{r}+\frac{1}{\varepsilon}}$  \\
     \cite{BKT20,Wan24} 
\end{tabular}         & \begin{tabular}{c}
     $\Omega(\frac{\sqrt{r}}{\varepsilon})$  \\
     \cref{thm:lb-td} 
\end{tabular}           & \textbf{New}        & Not Yet      \\ \midrule
Fidelity Estimation       & \begin{tabular}{c}
     $\Omega\rbra{r^{1/3}+\frac{1}{\varepsilon}}$  \\
     \cite{UNWT25} 
\end{tabular}         & \begin{tabular}{c}
     $\Omega(\frac{\sqrt{r}}{\sqrt{\varepsilon}}+\frac{1}{\varepsilon})$  \\
     \cref{thm:lb-fidelity} 
\end{tabular}           & \textbf{New}        & Not Yet      \\ \midrule
\begin{tabular}{c} Well-Conditioned \\ Fidelity Estimation \end{tabular}      & \begin{tabular}{c}
     $\Omega\rbra{\frac{1}{\varepsilon}}$  \\
     \cite{LWWZ25} 
\end{tabular}         & \begin{tabular}{c}
     $\Omega\rbra{\frac{1}{\varepsilon}}$  \\
     \cref{thm:lb-fidelity-kappa} 
\end{tabular}           & Reproved        & Near-\textbf{Optimal}      \\ 
\bottomrule
\end{tabular}
}
\end{table}

\paragraph{Other quantum tasks.}
In \cref{tab:other}, $8$ quantum query lower bounds for other quantum tasks are listed, including amplitude estimation, Hamiltonian simulation, quantum Gibbs sampling, and the entanglement entropy problem, where $5$ are new and $5$ are near-optimal. 

In particular, we present
\begin{enumerate}
    \item New proofs of a series of basic problems such as amplitude estimation (\cref{thm:lb_amp_est}), Hamiltonian simulation (\cref{thm:hamiltonian}), and quantum Gibbs sampling (\cref{thm:gibbs}), whereas the previous matching lower bounds are based on the quantum polynomial method \cite{NW99,BACS07} or unitary discrimination \cite{CKBG23,Weg25}. 
    \item The \textit{first} almost matching quantum query lower bound for the entanglement entropy problem originally considered in \cite{SY23} (see \cref{thm:ent-entropy-prob}), significantly improving the previous lower bounds by \cite{SY23,WZ25b,Weg25,CWZ24} and showing that the quantum algorithm given in \cite{Weg25} is near-optimal in $\Delta$. 
\end{enumerate}

\begin{table}[!htp]
\centering
\caption{Quantum query lower bounds for other quantum tasks.} \label{tab:other}
\adjustbox{max width=\textwidth}{
\begin{tabular}{ccccc}
\toprule
Problem  & Previous & This Paper & New?       & Optimal?     \\
\midrule
Amplitude Estimation      & \begin{tabular}{c}
     $\Omega\rbra{\frac{1}{\varepsilon}}$  \\
     \cite{NW99} 
\end{tabular}         & \begin{tabular}{c}
     $\Omega\rbra{\frac{1}{\varepsilon}}$  \\
     \cref{thm:lb_amp_est} 
\end{tabular}           & Reproved        & \textbf{Optimal}      \\ \midrule
Hamiltonian Simulation      & \begin{tabular}{c}
     $\Omega\rbra{t}$  \\
     \cite{BACS07} 
\end{tabular}         & \begin{tabular}{c}
     $\Omega\rbra{t}$  \\
     \cref{thm:hamiltonian} 
\end{tabular}           & Reproved        & \textbf{Optimal}      \\ \midrule
Gibbs Sampling      & \begin{tabular}{c}
     $\Omega\rbra{\beta}$  \\
     \cite{CKBG23,Weg25} 
\end{tabular}         & \begin{tabular}{c}
     $\Omega\rbra{\beta}$  \\
     \cref{thm:gibbs} 
\end{tabular}           & Reproved        & \textbf{Optimal}      \\ 
\begin{tabular}{c} Gibbs Sampling \\ (Stronger Oracle) \end{tabular}    & \begin{tabular}{c}
     $\widetilde{\Omega}\rbra{\beta}$  \\
     \cite{WZ25b} 
\end{tabular}         & \begin{tabular}{c}
     $\Omega\rbra{\beta}$  \\
     \cref{thm:stronger-gibbs} 
\end{tabular}           & \textbf{New}        & \textbf{Optimal}      \\ \midrule
\begin{tabular}{c} Entanglement Entropy \\ Problem ($0 < \alpha < 1$) \end{tabular}     & \begin{tabular}{c}
     $\widetilde{\Omega}\rbra{\frac{\sqrt{d}}{\sqrt{\Delta}}+\frac{d^{\frac{1}{2\alpha}-\frac{1}{2}}}{\Delta^{\frac{1}{2\alpha}}}}$  \\
     \cite{CWZ24} 
\end{tabular}         & \begin{tabular}{c}
     $\Omega\rbra{d^{\frac{1}{2\alpha}-o\rbra{1}}+\frac{d^{\frac{1}{2\alpha}-\frac{1}{2}}}{\Delta^{\frac{1}{2\alpha}}}+\frac{\sqrt{d}}{\sqrt{\Delta}}}$  \\
     \cref{thm:ent-entropy-prob} 
\end{tabular}           & \textbf{New}        & Not Yet      \\ 
\begin{tabular}{c} Entanglement Entropy \\ Problem ($\alpha = 1$) \end{tabular}     & \begin{tabular}{c}
     $\widetilde{\Omega}\rbra{\frac{\sqrt{d}}{\sqrt{\Delta}}}$  \\
     \cite{CWZ24} 
\end{tabular}         & \begin{tabular}{c}
     $\Omega\rbra{\frac{\sqrt{d}}{\sqrt{\Delta}}+\frac{\log\rbra{d}}{\Delta \sqrt{\log\rbra{\frac{1}{\Delta}}\log\log\rbra{\frac{1}{\Delta}}}}}$  \\
     \cref{thm:ent-entropy-prob} 
\end{tabular}           & \textbf{New}        & Not Yet      \\
\begin{tabular}{c} Entanglement Entropy \\ Problem ($\alpha > 1$) \end{tabular}     & \begin{tabular}{c}
     $\widetilde{\Omega}\rbra{\frac{\sqrt{d}}{\sqrt{\Delta}}}$  \\
     \cite{CWZ24} 
\end{tabular}         & \begin{tabular}{c}
     $\Omega\rbra{\frac{\sqrt{d}}{\sqrt{\Delta}}+\frac{d^{\frac{1}{2}-\frac{1}{2\alpha}}}{\Delta \sqrt{\log\rbra{\frac{1}{\Delta}}\log\log\rbra{\frac{1}{\Delta}}}}}$  \\
     \cref{thm:ent-entropy-prob} 
\end{tabular}           & \textbf{New}        & Not Yet      \\ 
\begin{tabular}{c} Entanglement Entropy \\ Problem (Integer $\alpha$) \end{tabular}     & \begin{tabular}{c}
     $\widetilde{\Omega}\rbra{\frac{\sqrt{d}}{\sqrt{\Delta}}}$  \\
     \cite{CWZ24} 
\end{tabular}         & \begin{tabular}{c}
     $\Omega\rbra{\frac{\frac{d^{1-\frac{1}{\alpha}}}{\Delta^{\frac{1}{\alpha}}}+\frac{d^{\frac{1}{2}-\frac{1}{2\alpha}}}{\Delta}}{\sqrt{\log\rbra{\frac{1}{\Delta}}\log\log\rbra{\frac{1}{\Delta}}}}}$  \\
     \cref{thm:ent-entropy-prob} 
\end{tabular}           & \textbf{New}        & \begin{tabular}{c} Near-\textbf{Optimal} \\ in $\Delta$ \end{tabular}      \\ 
\bottomrule
\end{tabular}
}
\end{table}

\subsubsection{Quantum sample upper bounds}

In \cref{tab:sample}, $10$ quantum sample upper bounds are listed, including entropy estimation and closeness estimation of quantum states, where all upper bounds are new and $1$ upper bound is known to be near-optimal. 

In particular, we present
\begin{enumerate}
    \item The state-of-the-art quantum algorithm for estimating the von Neumann entropy (\cref{thm:von-upper}), trace distance (\cref{thm:qs-td}) and fidelity (\cref{thm:qs-fidelity-rank}) of quantum states of rank $r$ with sample complexity $\widetilde{O}\rbra{r^2}$, which we conjecture are optimal in $r$ (see \cref{conj:s-vN}). 
    \item A quantum algorithm for well-conditioned fidelity estimation with sample complexity $\widetilde{O}\rbra{\frac{1}{\varepsilon^2}}$ (see \cref{thm:kappa-fidelity}), matching the lower bound given in \cite[Lemma 40]{LWWZ25}. 
\end{enumerate}

\begin{table}[t]
\centering
\caption{Quantum sample upper bounds.} \label{tab:sample}
\adjustbox{max width=\textwidth}{
\begin{tabular}{ccccc}
\toprule
Problem  & Previous & This Paper & New?       & Optimal?     \\
\midrule
\begin{tabular}{c}
     Von Neumann \\
     Entropy Estimation
\end{tabular}       & \begin{tabular}{c}
     $O\rbra{\frac{d^2}{\varepsilon}+\frac{\log^2\rbra{d}}{\varepsilon^2}}$  \\
     \cite{BMW16} 
\end{tabular}         & \begin{tabular}{c}
     $\widetilde{O}\rbra{\frac{r^2}{\varepsilon^4}}$  \\
     \cref{thm:von-upper} 
\end{tabular}           & \textbf{New}        & Not Yet      \\ \midrule
\begin{tabular}{c}
     $\alpha$-R\'enyi Entropy \\
     Estimation ($0 < \alpha < 1$)
\end{tabular}       & \begin{tabular}{c}
     $O\rbra{\frac{d^{\frac{2}{\alpha}}}{\varepsilon^{\frac{2}{\alpha}}}}$  \\
     \cite{AISW20} 
\end{tabular}         & \begin{tabular}{c}
     $\widetilde{O}\rbra{\frac{r^{\frac{2}{\alpha}}}{\varepsilon^{\frac{2}{\alpha}+2}}}$  \\
     \cref{thm:renyi-upper} 
\end{tabular}           & \textbf{New}        & Not Yet      \\ 
\begin{tabular}{c}
     $\alpha$-R\'enyi Entropy \\
     Estimation ($\alpha > 1$)
\end{tabular}       & \begin{tabular}{c}
     $O\rbra{\frac{d^2}{\varepsilon^2}}$  \\
     \cite{AISW20} 
\end{tabular}         & \begin{tabular}{c}
     $\widetilde{O}\rbra{\frac{r^2}{\varepsilon^{2+\frac{2}{\alpha}}}}$  \\
     \cref{thm:renyi-upper} 
\end{tabular}           & \textbf{New}        & Not Yet      \\ \midrule
\begin{tabular}{c}
     Trace Distance \\
     Estimation
\end{tabular}       & \begin{tabular}{c}
     $O\rbra{\frac{r^2}{\varepsilon^5}\log^2\rbra{\frac{r}{\varepsilon}}\log^2\rbra{\frac{1}{\varepsilon}}}$  \\
     \cite{WZ24} 
\end{tabular}         & \begin{tabular}{c}
     $O\rbra{\frac{r^2}{\varepsilon^4}\log^2\rbra{\frac{1}{\varepsilon}}}$  \\
     \cref{thm:qs-td} 
\end{tabular}           & \textbf{New}        & Not Yet      \\ \midrule
Fidelity Estimation       & \begin{tabular}{c}
     $\widetilde{O}\rbra{\frac{r^{5.5}}{\varepsilon^{12}}}$  \\
     \cite{GP22} 
\end{tabular}         & \begin{tabular}{c}
     $O\rbra{\frac{r^2}{\varepsilon^4}\log^2\rbra{\frac{1}{\varepsilon}}}$  \\
     \cite{UNWT25} \\
     \cref{thm:qs-fidelity-rank}
\end{tabular}           & \textbf{New}        & Not Yet      \\
\begin{tabular}{c}
     Well-Conditioned \\
     Fidelity Estimation
\end{tabular}       & \begin{tabular}{c}
     $\widetilde{O}\rbra{\frac{\kappa^9}{\varepsilon^3}}$  \\
     \cite{LWWZ25} 
\end{tabular}         & \begin{tabular}{c}
     $O\rbra{\frac{\kappa^2}{\varepsilon^2}\log^2\rbra{\frac{1}{\varepsilon}}}$  \\
     \cite{UNWT25} \\
     \cref{thm:kappa-fidelity}
\end{tabular}           & \textbf{New}        & \begin{tabular}{c} Near-\textbf{Optimal} \\ in $\varepsilon$ \end{tabular}     \\ \midrule
\begin{tabular}{c}
     Quantum $\ell_\alpha$ \\
     Distance Estimation
\end{tabular}       & \begin{tabular}{c}
     $\widetilde{O}\rbra{\frac{1}{\varepsilon^{3\alpha+2+\frac{2}{\alpha-1}}}}$  \\
     \cite{LW25b} 
\end{tabular}         & \begin{tabular}{c}
     $O\rbra{\frac{1}{\varepsilon^{2\alpha+2+\frac{2}{\alpha-1}}}}$  \\
     \cref{thm:lalpha} 
\end{tabular}           & \textbf{New}        & Not Yet      \\ \midrule
\begin{tabular}{c}
     $\alpha$-Tsallis Relative \\
     Entropy Estimation \\
     ($0 < \alpha < \frac{1}{2}$)
\end{tabular}       & \begin{tabular}{c}
     $O\rbra{\frac{r^{2+3\alpha}}{\varepsilon^{\frac{2}{\alpha} + \frac{3}{1-\alpha}}}\log^{2}\rbra{\frac{r}{\varepsilon}}}$  \\
     \cite{BGW25} 
\end{tabular}         & \begin{tabular}{c}
     $O\rbra{\frac{r^{2+2\alpha}}{\varepsilon^{\frac{2}{\alpha} + \frac{2}{1-\alpha}}}}$  \\
     \cref{thm:tsallis-relative-upper} 
\end{tabular}           & \textbf{New}        & Not Yet      \\ 
\begin{tabular}{c}
     $\alpha$-Tsallis Relative \\
     Entropy Estimation \\
     ($\frac{1}{2} < \alpha < 1$)
\end{tabular}       & \begin{tabular}{c}
     $O\rbra{\frac{r^{5-3\alpha}}{\varepsilon^{\frac{2}{1-\alpha}+ \frac{3}{\alpha}}}\log^{2}\rbra{\frac{r}{\varepsilon}}}$  \\
     \cite{BGW25} 
\end{tabular}         & \begin{tabular}{c}
     $O\rbra{\frac{r^{4-2\alpha}}{\epsilon^{\frac{2}{1-\alpha}+ \frac{2}{\alpha}}}}$  \\
     \cref{thm:tsallis-relative-upper} 
\end{tabular}           & \textbf{New}        & Not Yet      \\ 
\begin{tabular}{c}
     Squared Hellinger \\
     Distance Estimation
\end{tabular}       & \begin{tabular}{c}
     $O\rbra{\frac{r^{3.5}}{\varepsilon^{10}}\log^4\rbra{\frac{r}{\varepsilon}}}$  \\
     \cite{BGW25} 
\end{tabular}         & \begin{tabular}{c}
     $O\rbra{\frac{r^3}{\varepsilon^8}\log^2\rbra{\frac{r}{\varepsilon}}}$  \\
     \cref{thm:tsallis-relative-upper} 
\end{tabular}           & \textbf{New}        & Not Yet      \\
\bottomrule
\end{tabular}
}
\end{table}

\subsection{Organization of this paper}

In the remainder of this paper, quantum query lower bounds obtained by quantum sample-to-query lifting are listed in \cref{sec:lower} and quantum sample upper bounds obtained by quantum sample-to-query lifting are listed in \cref{sec:upper}. 
Open questions and conjectures are discussed in \cref{sec:discussion}. 

\section{List of Quantum Query Lower Bounds} \label{sec:lower}

\subsection{For classical properties}

The study of quantum testing of classical properties was initiated in \cite{BHH11}. 
For testing the properties of a $d$-dimensional classical distribution $P = \rbra{p_1, p_2, \dots, p_d}$, we assume a quantum query oracle $U$ acting on two subsystems $\mathsf{A}$ and $\mathsf{B}$ such that
\[
U_{\mathsf{AB}} \ket{0}_{\mathsf{A}} \ket{0}_{\mathsf{B}} = \sum_{i=1}^d \sqrt{p_i} \ket{i}_{\mathsf{A}} \ket{\psi_i}_{\mathsf{B}},
\]
where $\cbra{\ket{\psi_i}}_{i=1}^d$ is an orthonormal basis.\footnote{There have been several different quantum input models for testing classical properties. 
Here, we consider the purified quantum query access model, which is the weakest one (see \cite{Bel19,GL20} for the discussions and comparison with other models).} 
In other words, $U$ is a purified quantum query access oracle for the diagonal quantum state
\[
\rho = \sum_{i=1}^d p_i \ketbra{i}{i} = \diag\rbra{p_1, p_2, \dots, p_d}.
\]
For any promise problem $\mathcal{Q} = \cbra{\mathcal{Q}^{\textup{yes}}, \mathcal{Q}^{\textup{no}}}$ for classical distributions, one can define a promise problem $\mathcal{P} = \cbra{\mathcal{P}^{\textup{yes}}, \mathcal{P}^{\textup{no}}}$ for quantum state testing such that
\[
\mathcal{P}^X = \set{ \rho = \diag\rbra{p_1, p_2, \dots, p_d} }{ P = \rbra{p_1, p_2, \dots, p_d} \in \mathcal{Q}^X } \textup{ for } X \in \cbra{\textup{yes}, \textup{no}}. 
\]
It is clear that $\mathsf{S}\rbra{\mathcal{P}} = \mathsf{S}\rbra{\mathcal{Q}}$ and $\mathsf{Q}\rbra{\mathcal{P}} = \mathsf{Q}\rbra{\mathcal{Q}}$. 
By \cref{thm:main}, we immediately have 
\[
    \mathsf{Q}\rbra{\mathcal{Q}} = \mathsf{Q}\rbra{\mathcal{P}} = \Omega\rbra*{\sqrt{\mathsf{S}\rbra{\mathcal{P}}}} = \Omega\rbra*{\sqrt{\mathsf{S}\rbra{\mathcal{Q}}}},
\]
which allows us to establish a quantum query lower bound for $\mathcal{Q}$ from its classical sample lower bound. 

In the remaining part of this section, we list the quantum query lower bounds for testing properties of discrete distributions, where we will use $d$ to denote the dimension of the distributions. 

\subsubsection{Hypothesis testing}

\begin{problem} [Hypothesis testing] \label{prob:hypothesis-testing}
    Given an unknown distribution $R$ that is either $P$ or $Q$, determine whether $R$ is $P$ or $Q$.
\end{problem} 

The quantum query complexity of \cref{prob:hypothesis-testing} has been shown to be $\Theta\rbra{1/d_{\textup{H}}\rbra{P, Q}}$ in \cite{Bel19}. 
The original proof of the lower bound $\Omega\rbra{1/d_{\textup{H}}\rbra{P, Q}}$ given in \cite{Bel19} uses the adversary bound in \cite{BBH+18} (which is a generalization of the primal version of the general adversary bound \cite{HLS07}). 

By quantum sample-to-query lifting, we can simply reproduce this result.

\begin{theorem}[{\cite[Theorem 4]{Bel19}} reproduced] \label{thm:Bel19-reproduced}
    The quantum query complexity of hypothesis testing for two distributions $P$ and $Q$ is $\Omega\rbra{1/d_{\textup{H}}\rbra{P, Q}}$. 
\end{theorem}
\begin{proof}
    To apply \cref{thm:main}, we only have to note that the sample lower bound for hypothesis testing is $\Omega\rbra{1/d_{\textup{H}}^2\rbra{P, Q}}$ \cite{BYKS01} (see also \cite[Theorem 4.7]{BY02}). 
\end{proof}

\subsubsection{Closeness testing}

\begin{problem} [$\ell_\alpha$ closeness testing]
    Given an unknown distribution $P$ and a known distribution $Q$, determine whether $P = Q$ or $\Abs{P - Q}_{\alpha} \geq \varepsilon$, where the $\ell_\alpha$ distance is induced by the $\alpha$-Schattern norm $\Abs{P}_\alpha = \rbra{\sum_{i=1}^d p_i^\alpha}^{1/\alpha}$. In particular, the $\ell_1$ distance is called the total variation distance. 
\end{problem}

\paragraph{The case of $\alpha = 1$.}
The $\ell_1$ closeness testing was considered in \cite{GL20} with a quantum algorithm with query complexity $O\rbra{\frac{\sqrt{d}}{\varepsilon}\log^3\rbra{\frac{d}{\varepsilon}}\log\log\rbra{\frac{d}{\varepsilon}}}$. 
The current best quantum query upper bound for $\ell_1$ closeness testing is $O\rbra{\min\cbra{\frac{d^{1/3}}{\varepsilon^{4/3}}, \frac{\sqrt{d}}{\varepsilon}}}$, where $O\rbra{\frac{d^{1/3}}{\varepsilon^{4/3}}}$ is due to \cite{CKO25} and $O\rbra{\frac{\sqrt{d}}{\varepsilon}}$ is due to \cite{LWL24}.
The prior best quantum query lower bound is $\Omega\rbra{d^{1/3}+\frac{1}{\varepsilon}}$, where $\Omega\rbra{d^{1/3}}$ due to \cite{CFMdW10} shows the optimality in $d$ of \cite{CKO25} and $\Omega\rbra{\frac{1}{\varepsilon}}$ due to \cite{LWL24} shows the optimality in $\varepsilon$ of \cite{LWL24}. 

By quantum sample-to-query lifting, we can obtain an improved quantum query lower bound.
\begin{theorem}\label{thm:ell1-improved}
    The quantum query complexity of $\ell_1$ closeness testing is $\Omega\rbra{\frac{d^{1/3}}{\varepsilon^{2/3}}+\frac{d^{1/4}}{\varepsilon}}$. 
\end{theorem}
\begin{proof}
    To apply \cref{thm:main}, we only have note that the sample lower bound for $\ell_1$ closeness testing is $\Omega\rbra{\frac{d^{2/3}}{\varepsilon^{4/3}}+\frac{d^{1/2}}{\varepsilon^2}}$ \cite{CDVV14}. 
\end{proof}

For the special case where $Q$ is the uniform distribution, the problem is called \textit{uniformity testing} and quantum algorithms for it were given in \cite{BHH11} with query complexity $O\rbra{d^{1/3}}$ for $\varepsilon = \Theta\rbra{1}$ and in \cite{CFMdW10} with query complexity $O\rbra{\frac{d^{1/3}}{\varepsilon^2}}$. 
However, the current best quantum query upper bound remains $O\rbra{\min\cbra{\frac{d^{1/3}}{\varepsilon^{4/3}}, \frac{\sqrt{d}}{\varepsilon}}}$ due to \cite{CKO25,LWL24}; the prior best quantum query lower bound is $\Omega\rbra{d^{1/3}}$ due to \cite{CFMdW10} for $\varepsilon = \Theta\rbra{1}$. 

By quantum sample-to-query lifting, we can obtain an improved quantum query lower bound.
\begin{theorem}  \label{thm:lb-uniformity}
    The quantum query complexity of uniformity testing is $\Omega\rbra{d^{1/3}+\frac{d^{1/4}}{\varepsilon}}$. 
\end{theorem}
\begin{proof}
    To apply \cref{thm:main} to show a lower bound of $\Omega\rbra{\frac{d^{1/4}}{\varepsilon}}$, we only have to note that the sample lower bound for uniformity testing is $\Omega\rbra{\frac{\sqrt{d}}{\varepsilon^2}}$ \cite{Pan08,CDVV14}. 
\end{proof}
\cref{thm:lb-uniformity} shows that the current quantum query upper bound $O\rbra{\min\cbra{\frac{d^{1/3}}{\varepsilon^{4/3}}, \frac{\sqrt{d}}{\varepsilon}}}$ due to \cite{CKO25,LWL24} is optimal in both $d$ and $\varepsilon$. 

\paragraph{The case of $\alpha = 2$.}

A quantum algorithm for $\ell_2$ closeness testing was proposed in \cite{GL20} with query complexity $O\rbra{\frac{1}{\varepsilon}\log^3\rbra{\frac{1}{\varepsilon}}\log\log\rbra{\frac{1}{\varepsilon}}}$, which was later improved to $O\rbra{\frac{1}{\varepsilon}}$ in \cite{LWL24} together with a matching lower bound of $\Omega\rbra{\frac{1}{\varepsilon}}$. 
The lower bound in \cite{LWL24} was obtained by applying the lower bound in \cite{Bel19} for hypothesis testing. 

By quantum sample-to-query lifting, we can simply reproduce this result.

\begin{theorem}[{\cite[Theorem 13]{LWL24}} reproduced] \label{thm:ell2-rep}
    The quantum query complexity of $\ell_2$ closeness testing is $\Omega\rbra{\frac{1}{\varepsilon}}$. 
\end{theorem}
\begin{proof}
    To apply \cref{thm:main}, we only have to note that the sample lower bound for $\ell_2$ closeness testing is $\Omega\rbra{\frac{1}{\varepsilon^2}}$ \cite{CDVV14}. 
\end{proof}

\subsubsection{Distance estimation}

\begin{problem} [Total variation distance estimation]
    Given two unknown distributions $P$ and $Q$, estimate the $\ell_1$ distance $\Abs{P - Q}_1$ to within additive error $\varepsilon$. 
\end{problem}

A quantum algorithm for the $\ell_1$ distance estimation was proposed in \cite{BHH11} with query complexity $O\rbra{\frac{\sqrt{d}}{\varepsilon^8}}$.
The current best quantum query upper bound is $O\rbra{\frac{\sqrt{d}}{\varepsilon^2}\log\rbra{\frac{1}{\varepsilon}}}$ due to \cite{Mon15}.
The prior best quantum query lower bound is $\widetilde{\Omega}\rbra{\sqrt{d}}$ due to \cite[Theorem 1.5]{BKT20}, where the polylogarithmic factor in $d$ was omitted. 

By quantum sample-to-query lifting, we can establish a quantum query lower bound with (1) an explicit (and possibly improved) polylogarithmic factor in $d$ and (2) dependence on $\varepsilon$.

\begin{theorem} \label{thm:tvd-estimation}
    The quantum query complexity of $\ell_1$ distance estimation is $\Omega\rbra{\frac{\sqrt{d}}{\varepsilon\sqrt{\log\rbra{d/\varepsilon}}}}$. 
\end{theorem}
\begin{proof}
    To apply \cref{thm:main}, we only have to note that the sample lower bound for $\ell_1$ distance estimation is $\Omega\rbra{\frac{d}{\varepsilon^2\log\rbra{d/\varepsilon}}}$ \cite[Theorem 3]{JHW18}. 
\end{proof}

\subsubsection{Entropy estimation} \label{sec:c-entropy}

\begin{problem} [Entropy estimation]
    Given an unknown distribution $P$, estimate the entropy of $P$ to within additive error $\varepsilon$. 
\end{problem}

\paragraph{Shannon entropy.}
The Shannon entropy is defined by $\mathrm{H}\rbra{P} = -\sum_{i=1}^d p_i \ln\rbra{p_i}$. 
A quantum algorithm for Shannon entropy estimation was proposed in \cite[Theorem 6]{LW19} with query complexity $\widetilde{O}\rbra{\frac{\sqrt{d}}{\varepsilon^2}}$, which was later improved to $O\rbra{\frac{\sqrt{d}}{\varepsilon^{1.5}}\rbra{\frac{d}{\varepsilon}}\log\log\rbra{\frac{\log\rbra{d}}{\varepsilon}}}$ in \cite[Theorem 12]{GL20}.
The current best quantum query upper bound is $\widetilde{O}\rbra{\frac{\sqrt{d}}{\varepsilon}}$ due to \cite{SJ25} and the prior best quantum query lower bound is $\widetilde{\Omega}\rbra{\sqrt{d}}$ due to \cite[Theorem 1.6]{BKT20}, where the polylogarithmic factor in $d$ was omitted. 

By quantum sample-to-query lifting, we can establish a quantum query lower bound with (1) an explicit (and possibly improved) polylogarithmic factor in $d$ and (2) optimal dependence on $\varepsilon$.

\begin{theorem} \label{thm:lb-shannon}
    The quantum query complexity of Shannon entropy estimation is $\Omega\rbra{\frac{\sqrt{d}}{\sqrt{\varepsilon \log\rbra{d}}}+\frac{\log\rbra{d}}{\varepsilon}}$. 
\end{theorem}
\begin{proof}
    To apply \cref{thm:main}, we only have to note that the sample lower bound for Shannon entropy estimation is $\Omega\rbra{\frac{d}{\varepsilon \log\rbra{d}}+\frac{\log^2\rbra{d}}{\varepsilon^2}}$ \cite{JVHW15,WY16}. 
\end{proof}

\cref{thm:lb-shannon} shows that the quantum query upper bound given in \cite{SJ25} is optimal in both $d$ and $\varepsilon$. 

\paragraph{R\'enyi entropy.}
The $\alpha$-R\'enyi entropy is defined by $\mathrm{H}_{\alpha}^{\textup{R\'en}}\rbra{P} = \frac{1}{1-\alpha}\ln\rbra{\sum_{i=1}^d p_i^\alpha}$. 

\textit{The case of $0 < \alpha < 1$}. A quantum algorithm for estimating the $\alpha$-R\'enyi entropy was proposed in \cite{LW19} with query complexity $\widetilde{O}\rbra{\frac{d^{\frac{1}{\alpha}-\frac{1}{2}}}{\varepsilon^{\frac{1}{2\alpha}+1}}}$ (this is the corrected one noted in \cite{WZL24}).
The current best quantum query upper bound is $\widetilde{O}\rbra{\frac{d^{\frac{1}{2\alpha}}}{\varepsilon^{\frac{1}{2\alpha}+1}}}$ due to \cite{WZL24} and the prior best quantum query lower bound is (1) $\Omega\rbra{\frac{d^{\frac{1}{2\alpha}-\frac{1}{2}}}{\varepsilon^{\frac{1}{2\alpha}}}}$ due to \cite[Theorem 9]{WZL24} and (2) $\Omega\rbra{\frac{d^{\frac{1}{3}}}{\varepsilon^{\frac{1}{6}}}}$ when $\frac{1}{d} \leq \varepsilon \leq \frac{1}{2}$ due to \cite[Theorem 14]{LW19}. 

By quantum sample-to-query lifting, we can establish a matching quantum query lower bound below.
\begin{theorem} \label{thm:lt1-renyi}
    For $0 < \alpha < 1$, the quantum query complexity of $\alpha$-R\'enyi entropy estimation is $\Omega\rbra{d^{\frac{1}{2\alpha}-o\rbra{1}}+\frac{d^{\frac{1}{2\alpha}-\frac{1}{2}}}{\varepsilon^{\frac{1}{2\alpha}}}}$. 
\end{theorem}
\begin{proof}
    To apply \cref{thm:main} to show a lower bound of $\Omega\rbra{d^{\frac{1}{2\alpha}-o\rbra{1}}}$, we only have to note that the sample lower bound for $\alpha$-R\'enyi entropy estimation is $\Omega\rbra{d^{\frac{1}{\alpha}-o\rbra{1}}}$ in \cite[Theorem 22]{AOST17}. 
\end{proof}
\cref{thm:lt1-renyi} shows that the quantum algorithm in \cite{WZL24} for $0 < \alpha < 1$ has near-optimal dependence on $d$. 

\textit{The case of non-integer $\alpha > 1$}. 
A quantum algorithm for estimating the $\alpha$-R\'enyi entropy was proposed in \cite{LW19} with query complexity 
$\widetilde{O}\rbra{\frac{d^{1-\frac{1}{2\alpha}}}{\varepsilon^2}}$. 
The current best quantum query upper bound is
$\widetilde{O}\rbra{\frac{d^{1-\frac{1}{2\alpha}}}{\varepsilon}+\frac{\sqrt{d}}{\varepsilon^{1+\frac{1}{2\alpha}}}}$ due to \cite{WZL24}, and the prior best quantum query lower bound is $\Omega\rbra{\frac{d^{\frac{1}{3}}}{\varepsilon^{\frac{1}{6}}}+\frac{d^{\frac{1}{2}-\frac{1}{2\alpha}}}{\varepsilon}}$ due to \cite{LW19}.

By quantum sample-to-query lifting, we can obtain an improved quantum query lower bound.

\begin{theorem} \label{thm:renyi-gt-1}
    For non-integer $\alpha > 1$, the quantum query complexity of $\alpha$-R\'enyi entropy estimation is $\Omega\rbra{\frac{d^{\frac{1}{2}-\frac{1}{2\alpha}}}{\varepsilon} + d^{\frac{1}{2}-o\rbra{1}}}$. 
\end{theorem}
\begin{proof}
    To apply \cref{thm:main} to show a lower bound of $\Omega\rbra{d^{\frac{1}{2}-o\rbra{1}}}$, we only have to note that the sample lower bound for $\alpha$-R\'enyi entropy estimation is $\Omega\rbra{d^{1-o\rbra{1}}}$ in \cite[Theorem 21]{AOST17}. 
\end{proof}

\textit{The case of integer $\alpha \geq 2$.}
A quantum algorithm for estimating the $\alpha$-R\'enyi entropy was proposed in \cite{LW19} with query complexity $\widetilde{O}\rbra{\frac{d^{\nu\rbra{1-\frac{1}{\alpha}}}}{\varepsilon^2}}$, where $\nu = 1 - \frac{2^{\alpha-2}}{2^{\alpha}-1} < \frac{3}{4}$. 
A matching quantum query lower bound $\widetilde{\Omega}\rbra{d^{1/3}}$ is given in \cite{LW19} for $\alpha = 2$. 
For integer $\alpha \geq 3$, the prior best quantum query lower bound is $\Omega\rbra{\frac{d^{\frac{1}{2}-\frac{1}{2\alpha}}}{\varepsilon}}$ due to \cite{LW19}.

By quantum sample-to-query lifting, we can simply reproduce this result.

\begin{theorem}[{\cite[Theorem 14]{LW19}} reproduced] \label{thm:integer-renyi}
    For integer $\alpha \geq 3$, the quantum query complexity of $\alpha$-R\'enyi entropy estimation is $\Omega\rbra{\frac{d^{\frac{1}{2}-\frac{1}{2\alpha}}}{\varepsilon}}$. 
\end{theorem}
\begin{proof}
    To apply \cref{thm:main}, we only have to note that the sample lower bound for $\alpha$-R\'enyi entropy estimation for integer $\alpha \geq 2$ is $\Omega\rbra{\frac{d^{1-\frac{1}{\alpha}}}{\varepsilon^2}}$ in \cite[Theorem 15]{AOST17}. 
\end{proof}

\paragraph{Tsallis entropy.}
The $\alpha$-Tsallis entropy is defined by $\mathrm{H}_{\alpha}^{\textup{Tsa}}\rbra{P} = \frac{1}{1-\alpha}\rbra{\sum_{i=1}^d p_i^\alpha-1}$. 

\textit{The case of integer $\alpha \geq 2$.}
A quantum algorithm for estimating the $\alpha$-Tsallis entropy with query complexity $O\rbra{\frac{\sqrt{\log\rbra{1/\alpha\varepsilon}}}{\sqrt{\alpha}\varepsilon}}$ was given in \cite{Wan25} together with a lower bound of $\Omega\rbra{\frac{1}{\sqrt{\alpha}\varepsilon}}$. 

By quantum sample-to-query lifting, we can simply reproduce this result.

\begin{theorem}[{\cite[Lemma 1.4]{Wan25}} reproduced] \label{thm:integer-tsallis}
    For integer $\alpha \geq 2$, the quantum query complexity of $\alpha$-Tsallis entropy estimation is $\Omega\rbra{\frac{1}{\sqrt{\alpha}\varepsilon}}$. 
\end{theorem}
\begin{proof}
    To apply \cref{thm:main}, we only have to note that the sample lower bound for $\alpha$-Tsallis entropy estimation for integer $\alpha \geq 2$ is $\Omega\rbra{\frac{1}{\alpha\varepsilon^2}}$ in \cite[Theorem 4.3]{CWYZ25}. 
\end{proof}

\subsubsection{Support size}

\begin{problem} [Support size estimation]
    Given an unknown distribution $P$ with minimum nonzero mass at least $\frac{1}{d}$, estimate the support size of $P$ to within additive error $\varepsilon d$. 
\end{problem}

The problem of support size estimation is also known as the $0$-R\'enyi entropy estimation, where the difference is that one has to assume the minimum nonzero mass of the unknown distribution. 
A quantum algorithm for support size estimation was given in \cite{LW19} with query complexity $\widetilde{O}\rbra{\frac{\sqrt{d}}{\varepsilon^{1.5}}}$, which is also the current best one. 
The prior best quantum query lower bound is $\widetilde{\Omega}\rbra{\sqrt{d}}$ due to \cite[Theorem 6.2]{BKT20}. 

By quantum sample-to-query lifting, we can obtain an improved quantum query lower bound.
\begin{theorem} \label{thm:support-size}
    The quantum query complexity of support size estimation is $\Omega\rbra{\frac{\sqrt{d}}{\sqrt{\log\rbra{d}}}\log\rbra{\frac{1}{\varepsilon}}}$. 
\end{theorem}
\begin{proof}
    To apply \cref{thm:main}, we only have to note that the sample lower bound for support size estimation is $\Omega\rbra{\frac{d}{\log\rbra{d}}\log^2\rbra{\frac{1}{\varepsilon}}}$ \cite{WY19}. 
\end{proof}

\subsubsection{\texorpdfstring{$k$}{k}-wise uniformity testing}

\begin{problem} [$k$-wise uniformity testing]
    Given an unknown distribution $P$, determine whether it is $k$-wise uniform or $\varepsilon$-far from any $k$-wise uniform distributions (in total variation distance). 
    Here, a distribution $P$ over $\cbra{0, 1}^n$ is said to be $k$-wise uniform, if for every subset $S \subseteq \cbra{1, 2, \dots, n}$ with $\abs{S} = k$, 
    \[
    \Pr_{x \sim P} \sbra{x_S = v} = \frac{1}{2^k}
    \]
    for any $v \in \cbra{0, 1}^k$, where $x_S$ means the subsequence of $x$ with index in $S$. 
\end{problem}

A quantum algorithm for $k$-wise uniformity testing was given in \cite{LWL24} with query complexity $O\rbra{\frac{n^{\frac{k}{2}}}{\varepsilon}}$ for constant $k \geq 2$. 

By quantum sample-to-query lifting, we can show a quantum query lower bound when $k = 2$.

\begin{theorem} \label{thm:k-wise}
    The quantum query complexity of $2$-wise uniformity testing is $\Omega\rbra{\frac{\sqrt{n}}{\varepsilon}}$. 
\end{theorem}
\begin{proof}
    To apply \cref{thm:main}, we only have to note that the sample lower bound for $2$-wise uniformity testing is $\Omega\rbra{\frac{n}{\varepsilon^2}}$ \cite[Theorem 5]{OZ18}.
\end{proof}

\cref{thm:k-wise} shows that the quantum query upper bound given in \cite{LWL24} is optimal in $\varepsilon$. 

\subsection{For quantum properties}

For the quantum query complexity of testing quantum properties, we assume purified quantum query access (\cref{def:purified-access}) to the input quantum states. 

\subsubsection{Mixedness testing}

\begin{problem}[Mixedness testing]
Given an unknown $d$-dimensional state $\rho$, determine whether it is the maximally mixed (i.e., $\rho=\frac{I}{d}$) or $\varepsilon$-far from the maximally mixed state in trace distance (i.e., $\|\rho-\frac{I}{d}\|_1\geq \varepsilon$). 
\end{problem}

\begin{theorem}\label{thm-10150001}
The quantum query complexity of mixedness testing is $\Omega(\frac{\sqrt{d}}{\varepsilon})$.
\end{theorem}
\begin{proof}
    To apply \cref{thm:main}, we only have to note that the sample lower bound for mixedness testing is $\Omega\rbra{\frac{d}{\varepsilon^2}}$ \cite[Theorem 1.10]{OW21}.
\end{proof}

The lower bound in \cref{thm-10150001} also applies to quantum state certification \cite{BOW19}, where the task is to determine whether an unknown quantum state $\rho$ is identical to a given one $\sigma$ or $\rho$ is $\varepsilon$-far from $\sigma$ in trace distance. 

\begin{corollary} \label{corollary:state-cert}
    The quantum query complexity of quantum state certification is $\Omega\rbra{\frac{\sqrt{d}}{\varepsilon}}$. 
\end{corollary}

Moreover, \cref{thm-10150001} also implies an improved quantum query lower bound given in \cite[Corollary 4.19]{WZ25b} for the mixedness testing of block-encoded matrices.

\begin{corollary} \label{corollary:mixedness-be}
    The quantum query complexity of determining whether an Hermitian matrix $A$ block-encoded in a unitary oracle has spectrum $\rbra{\frac{1}{d}, \dots, \frac{1}{d}}$ or $\varepsilon$-far in trace norm is $\Omega\rbra{\frac{\sqrt{d}}{\varepsilon}}$. 
\end{corollary}

\subsubsection{Rank testing}

\begin{problem}[Rank testing]
Given an unknown state $\rho$, determine whether it is of rank at most $r$ or $\varepsilon$-far from any state that is of rank at most $r$ in trace norm.
\end{problem}
\begin{theorem} \label{thm:rank-testing}
The quantum query complexity of rank testing is $\Omega(\frac{\sqrt{r}}{\sqrt{\varepsilon}})$.
\end{theorem}
\begin{proof}
    To apply \cref{thm:main}, we only have to note that the sample lower bound for rank testing is $\Omega\rbra{\frac{r}{\varepsilon}}$ \cite[Theorem 1.11]{OW21}.
\end{proof}

\cref{thm:rank-testing} also implies an improved quantum query lower bound given in \cite[Corollary 4.8]{WZ25b} for the rank testing of block-encoded matrices. 

\begin{corollary} \label{corollary:rank-be}
    The quantum query complexity of determining whether an Hermitian matrix $A$ block-encoded in a unitary oracle has rank $\leq r$ or $\varepsilon$-far in trace norm from any matrices of rank $\leq r$ is $\Omega\rbra{\frac{\sqrt{r}}{\sqrt{\varepsilon}}}$. 
\end{corollary}

\subsubsection{Uniformity distinguishing}

\begin{problem}[Uniformity distinguishing]
Given an unknown state $\rho$, determine whether $\rho$'s spectrum is uniform on either $r$ or $r+\Delta$ eigenvalues.
\end{problem}
\begin{theorem} \label{thm:uniformity-distinguishing}
For $1\leq \Delta< r$, the quantum query complexity of uniformity distinguishing is $\Omega^*(\frac{r}{\sqrt{\Delta}})$.\footnote{$\Omega^*\rbra{\cdot}$ suppresses quasi-polylogarithmic factors.} For $\Delta = r$, a better lower bound of $\Omega(\sqrt{r})$ can be obtained.
\end{theorem}
\begin{proof}
    To apply \cref{thm:main}, we only have to note that the sample lower bound for uniformity distinguishing is $\Omega^*\rbra{\frac{r^2}{\Delta}}$ \cite[Theorem 1.12]{OW21} for $1\leq \Delta<r$.
    For $\Delta=r$, a better query lower bound of $\Omega(\sqrt{r})$ can be obtained by noting the sample lower bound $\Omega\rbra{r}$ for uniformity distinguishing in \cite[Theorem 3]{CHW07} (see also \cite[Theorem 1.9]{OW21}).
\end{proof}

\cref{thm:uniformity-distinguishing} also implies an improved quantum query lower bound given in \cite[Corollary 4.10]{WZ25b} for the rank testing of block-encoded matrices when $\Delta = r$. 

\begin{corollary} \label{corollary:uniform-dis-be}
    The quantum query complexity of determining whether an Hermitian matrix block-encoded in a unitary oracle has spectrum uniform on either $r$ or $2r$ eigenvalues is $\Omega\rbra{\sqrt{r}}$. 
\end{corollary}

\subsubsection{Entanglement spectrum testing}
\begin{problem}[Maximal entanglement testing]
Given an unknown bipartite state $\ket{\psi}\in \mathbb{C}^d\otimes\mathbb{C}^d$, determine whether it is the maximally entangled or $\varepsilon$-far from any maximally entangled state in trace norm. 
\end{problem}

\begin{theorem} \label{thm:maximal-entanglement}
The quantum query complexity of maximal entanglement testing is $\Omega(\frac{\sqrt{d}}{\varepsilon})$.
\end{theorem}
\begin{proof}
   To apply \cref{thm:main}, we only have to note that the sample lower bound for maximal entanglement testing is $\Omega\rbra{\frac{d}{\varepsilon^2}}$~\cite[Corollary 1.6]{CWZ24}.
\end{proof}

\begin{problem}[Schmidt rank testing]
Given an unknown  bipartite state $\ket{\psi}$, determine whether it is of Schmidt rank at most $r$ or $\varepsilon$-far from any state that is of Schmidt rank at most $r$ in trace norm.
\end{problem}
\begin{theorem} \label{thm:schmidt-rank}
The quantum query complexity of Schmidt rank testing is $\Omega(\frac{\sqrt{r}}{\varepsilon})$.
\end{theorem}
\begin{proof}
   To apply \cref{thm:main}, we only have to note that the sample lower bound for Schmidt rank testing is $\Omega\rbra{\frac{r}{\varepsilon^2}}$~\cite[Corollary 1.4]{CWZ24}.
\end{proof}

\begin{problem}[Uniform Schmidt coefficient testing]
Given an unknown bipartite state $\ket{\psi}$, determine whether $\ket{\psi}$'s entanglement spectrum is uniform on either $r$ or $r+\Delta$ Schmidt coefficients.
\end{problem}
\begin{theorem}\label{thm:uniform-schmidt}
For $1\leq \Delta < r$, the quantum query complexity of uniform Schmidt coefficient testing is $\Omega^*(\frac{r}{\sqrt{\Delta}})$. For $\Delta = r$, a better lower bound of $\Omega(\sqrt{r})$ can be obtained.
\end{theorem}
\begin{proof}
   To apply \cref{thm:main}, we only have to note that the sample lower bound for uniform Schmidt coefficient testing is $\Omega^*\rbra{\frac{r^2}{\Delta}}$~\cite[Corollary 1.8]{CWZ24} for $1\leq \Delta<r$. 
   For $\Delta = r$, a better query lower bound of $\Omega\rbra{\sqrt{r}}$ can be obtained by noting the sample lower bound $\Omega\rbra{r}$ for uniform Schmidt coefficient testing implied by \cite[Corollary 1.3]{CWZ24} and \cite[Theorem 3]{CHW07} (see also \cite[Theorem 1.9]{OW21}). 
\end{proof}

\subsubsection{Quantum entropy estimation}
\begin{problem}[Quantum entropy estimation]
    Given an unknown quantum state $\rho$, estimate the entropy of $\rho$ to within additive error $\varepsilon$. 
\end{problem}
\paragraph{Von Neumann entropy.} Given a $d$-dimensional quantum state $\rho$, the von Neumann entropy of $\rho$ is defined by 
\[
\mathrm{S}(\rho)=-\tr(\rho\ln(\rho)).
\]
Quantum algorithms were given in \cite{GL20} with query complexity $\widetilde{O}\rbra{\frac{d}{\varepsilon^{1.5}}}$ and in \cite{WGL+24} with query complexity $\widetilde{O}\rbra{\frac{r}{\varepsilon^2}}$ when $\rho$ is of rank $r$. 
The prior best quantum query lower bound is $\widetilde{\Omega}\rbra{\sqrt{d}}$ due to \cite[Theorem 1.6]{BKT20}. 

By quantum sample-to-query lifting, we can obtain an improved quantum query lower bound.

\begin{theorem} \label{thm:von-neumann-entropy}
The quantum query complexity of von Neumann entropy is $\Omega\rbra{\frac{\sqrt{d}}{\sqrt{\varepsilon}}+\frac{\log\rbra{d}}{\varepsilon}}$.
\end{theorem}
\begin{proof}
Note that a sample lower bound of $\Omega(\frac{d}{\varepsilon})$ was shown in \cite[Theorem 1.7]{WZ25}. On the other hand, a sample lower bound of $\Omega(\frac{\log^2(d)}{\varepsilon^2})$ for estimating the Shannon entropy was shown in \cite{JVHW15,WY16}. By simulating the measurement of von Neumann entropy estimator on the diagonal state, one can obtain a Shannon entropy estimator with the same sample complexity. Thus the lower bound $\Omega(\frac{\log^2(d)}{\varepsilon^2})$ also applies to the von Neumann entropy. Therefore, we obtain a sample lower bound of $\Omega\rbra{\frac{d}{\varepsilon}+\frac{\log^2\rbra{d}}{\varepsilon^2}}$. 
Then, by \cref{thm:main}, we obtain the query lower bound 
$\Omega(\sqrt{\frac{d}{\varepsilon}+\frac{\log^2\rbra{d}}{\varepsilon^2}})=\Omega\rbra{\frac{\sqrt{d}}{\sqrt{\varepsilon}}+\frac{\log\rbra{d}}{\varepsilon}}$.
\end{proof}

\paragraph{Quantum R\'enyi entropy.} The quantum $\alpha$-R\'enyi entropy is defined by 
\[
\mathrm{S}_\alpha^{\textup{R\'en}}(\rho)=\frac{1}{1-\alpha}\ln\rbra*{\tr(\rho^\alpha)}.
\]
The quantum query complexity of quantum R\'enyi entropy was studied in \cite{SH21,WZW23,WZYW23,WGL+24}. 
The current best quantum query upper bound for quantum R\'enyi entropy estimation is $\widetilde{O}\rbra{\frac{d^{\frac{1}{2\alpha}+\frac{1}{2}}}{\varepsilon^{\frac{1}{2\alpha}+1}}}$ for $0 < \alpha < 1$ and $\widetilde{O}\rbra{\min\cbra{\frac{d^{\frac{3}{2}-\frac{1}{2\alpha}}}{\varepsilon} + \frac{d}{\varepsilon^{1+\frac{1}{2\alpha}}}, \frac{d}{\varepsilon^{1+\frac{1}{\alpha}}}}}$ for $\alpha > 1$ due to \cite{WZL24}. 
The current best quantum query lower bounds are inherited from those for the estimation of classical entropy as mentioned in \cref{sec:c-entropy}. 

By quantum sample-to-query lifting, we can obtain an improved quantum query lower bound.

\begin{theorem} \label{thm:qRenyi}
For $0<\alpha<1$, the query complexity of quantum $\alpha$-R\'enyi entropy is $\Omega\rbra{d^{\frac{1}{2\alpha}-o\rbra{1}}+\frac{d^{\frac{1}{2\alpha}-\frac{1}{2}}}{\varepsilon^{\frac{1}{2\alpha}}}+\frac{\sqrt{d}}{\sqrt{\varepsilon}}}$. 
For non-integer $\alpha>1$, the query complexity of quantum $\alpha$-R\'enyi entropy is $\Omega\rbra{\frac{d^{\frac{1}{2}-\frac{1}{2\alpha}}}{\varepsilon}+\frac{\sqrt{d}}{\sqrt{\varepsilon}}}$.
For integer $\alpha\geq 2$, the query complexity of quantum $\alpha$-R\'enyi entropy is $\Omega(\frac{d^{1-\frac{1}{\alpha}}}{\varepsilon^{\frac{1}{\alpha}}}+\frac{d^{\frac{1}{2}-\frac{1}{2\alpha}}}{\varepsilon})$.
\end{theorem}
\begin{proof}
For $0<\alpha<1$, \cref{thm:lt1-renyi} gives part of the query lower bound $\Omega\rbra{d^{\frac{1}{2\alpha}-o\rbra{1}}+\frac{d^{\frac{1}{2\alpha}-\frac{1}{2}}}{\varepsilon^{\frac{1}{2\alpha}}}}$. 
Moreover, a query lower bound of $\Omega\rbra{\frac{\sqrt{d}}{\sqrt{\varepsilon}}}$ is obtained by \cref{thm:main} and noting the sample lower bound $\Omega\rbra{\frac{d}{\varepsilon}}$ for quantum $\alpha$-R\'enyi entropy estimation given in \cite[Theorem VI.5]{WZ25}. 

For non-integer $\alpha>1$, \cref{thm:renyi-gt-1} gives part of the query lower bound $\Omega\rbra{\frac{d^{\frac{1}{2}-\frac{1}{2\alpha}}}{\varepsilon}}$.
Moreover, a query lower bound of $\Omega\rbra{\frac{\sqrt{d}}{\sqrt{\varepsilon}}}$ is obtained by \cref{thm:main} and noting the sample lower bound $\Omega\rbra{\frac{d}{\varepsilon}}$ for quantum $\alpha$-R\'enyi entropy estimation given in \cite[Theorem VI.4]{WZ25}. 

For integer $\alpha\geq 2$, the query lower bound is obtained by \cref{thm:main} and noting the sample lower bound $\Omega(\frac{d^{2-\frac{2}{\alpha}}}{\varepsilon^{\frac{2}{\alpha}}}+\frac{d^{1-\frac{1}{\alpha}}}{\varepsilon^2})$ given in \cite[Theorem 1]{AISW20}.
\end{proof}

\paragraph{Quantum Tsallis entropy.} The quantum $\alpha$-Tsalllis entropy is defined by 
\[\mathrm{S}_\alpha^{\textup{Tsa}}(\rho)=\frac{1}{1-\alpha}\rbra*{\tr(\rho^\alpha)-1}.\]
The current best quantum query complexity for quantum Tsallis entropy estimation is $\widetilde{O}\rbra{\frac{1}{\varepsilon^{1+\frac{1}{\alpha-1}}}}$ for $\alpha > 1$ due to \cite{LW25}, where they provided a quantum query lower bound of $\Omega\rbra{\frac{1}{\sqrt{\varepsilon}}}$. 

\begin{theorem} \label{thm:qTsallis}
    For $1 < \alpha < \frac{3}{2}$, the quantum query complexity of quantum $\alpha$-Tsallis entropy estimation is $\Omega\rbra{\frac{1}{\varepsilon^{\frac{1}{2\rbra{\alpha-1}}}}}$. 
    For $\alpha \geq \frac{3}{2}$, the quantum query complexity is $\Omega\rbra{\frac{1}{\varepsilon}}$. 
\end{theorem}
\begin{proof}
    To apply \cref{thm:main}, we only have to note that the sample lower bound for $\alpha$-Tsallis entropy estimation is $\Omega\rbra{\frac{1}{\varepsilon^{\frac{1}{\alpha-1}}}}$ for $1 < \alpha < \frac{3}{2}$ and $\Omega\rbra{\frac{1}{\varepsilon^2}}$ for $\alpha \geq \frac{3}{2}$ in \cite[Theorem 20]{CW25}. 
\end{proof}

\subsubsection{Trace distance estimation}

\begin{problem}[Trace distance estimation]
    Given two unknown quantum states $\rho$ and $\sigma$ of rank $r$, estimate the trace distance 
    \[\mathrm{T}(\rho,\sigma)=\frac{1}{2}\tr\!\left(|\rho-\sigma|\right)\]
    to within additive error $\varepsilon$. 
\end{problem}

A quantum algorithm for trace distance estimation was proposed in \cite[Theorem IV.1]{WGL+24} with query complexity $\widetilde{O}\rbra{\frac{r^{5}}{\varepsilon^6}}$, which was later improved to ${O}\rbra{\frac{r}{\varepsilon^2}\log\rbra{\frac{1}{\varepsilon}}}$ in \cite[Corollary III.2]{WZ24}. 
For the pure-state case where $r = 1$, optimal quantum algorithms for trace distance estimation were proposed in \cite{Wan24} with query complexity $\Theta\rbra{\frac{1}{\varepsilon}}$. 
The prior best quantum query lower bound is $\widetilde{\Omega}\rbra{\sqrt{r}+\frac{1}{\varepsilon}}$, where $\widetilde{\Omega}\rbra{\sqrt{r}}$ is due to the quantum query lower bound for total variation distance estimation \cite[Theorem 1.5]{BKT20} with polylogarithmic factors in $r$ omitted and $\Omega\rbra{\frac{1}{\varepsilon}}$ was mentioned in \cite[Theorem V.2]{Wan24}. 

By quantum sample-to-query lifting, we can obtain an improved quantum query lower bound.

\begin{theorem} \label{thm:lb-td}
The quantum query complexity of trace distance estimation is $\Omega(\frac{\sqrt{r}}{\varepsilon})$.
\end{theorem}
\begin{proof}
Given an $r$-dimensional state $\rho$, we can embed it into a $d$-dimensional Hilbert space using an isometry $V:\mathbb{C}^r\rightarrow\mathbb{C}^d$. Then, any trace estimation algorithm on $V\rho V^\dag$ and $\frac{1}{r}V I_r V^\dag$ ($\frac{1}{r}I_r$ is the $r$-dimensional maximally mixed state) can be used to test whether $\rho$ is the maximally mixed state or $\varepsilon$-far from the maximally mixed state in trace distance. Therefore, the sample lower bound $\Omega(\frac{r}{\varepsilon^2})$ in \cite[Theorem 1.10]{OW21} applies also to the trace estimation task. 
By \cref{thm:main}, we obtain a quantum query lower bound $\Omega(\frac{\sqrt{r}}{\varepsilon})$ for trace distance estimation.
\end{proof}

\subsubsection{Fidelity estimation}

\begin{problem}[Fidelity estimation]
    Given two unknown quantum states $\rho$ and $\sigma$ of rank $r$, estimate the fidelity 
    \[\mathrm{F}(\rho,\sigma)=\tr\!\left(\sqrt{\sqrt{\sigma}\rho\sqrt{\sigma}}\right)\]
    to within additive error $\varepsilon$. 
\end{problem}

A quantum algorithm for fidelity estimation was proposed in \cite{WZC+23} with query complexity $\widetilde{O}\rbra{\frac{r^{12.5}}{\varepsilon^{13.5}}}$.\footnote{An earlier version of \cite{WZC+23} gave a quantum query complexity of $\widetilde{O}\rbra{\frac{r^{21.5}}{\varepsilon^{23.5}}}$ for fidelity estimation.} 
This was later improved to $\widetilde{O}\rbra{\frac{r^{6.5}}{\varepsilon^{7.5}}}$ in \cite[Theorem IV.5]{WGL+24}, to $\widetilde{O}\rbra{\frac{r^{2.5}}{\varepsilon^{5}}}$ in \cite[Corollary 28]{GP22}, and recently to $\widetilde{O}\rbra{\frac{r}{\varepsilon^2}\log\rbra{\frac{1}{\varepsilon}}}$ in \cite[Theorem 23]{UNWT25}. 
For the pure-state case where $r = 1$, optimal quantum algorithms for fidelity estimation were proposed in \cite{Wan24,FW25} with query complexity $\Theta\rbra{\frac{1}{\varepsilon}}$. 
The prior best quantum query lower bound is $\Omega\rbra{r^{1/3}+\frac{1}{\varepsilon}}$ mentioned in \cite[Proposition 31]{UNWT25}, where $\Omega\rbra{r^{1/3}}$ is due to the quantum query lower bound for uniformity testing of probability distributions \cite{CFMdW10} and $\Omega\rbra{1/\varepsilon}$ is due to the lower bound for quantum counting \cite{BBC+01,NW99}. 

By quantum sample-to-query lifting, we can obtain an improved quantum query lower bound.

\begin{theorem} \label{thm:lb-fidelity}
The quantum query complexity of fidelity estimation is $\Omega(\frac{\sqrt{r}}{\sqrt{\varepsilon}}+\frac{1}{\varepsilon})$.
\end{theorem}
\begin{proof}
Given an $r$-dimensional state $\rho$, we can embed it into a $d$-dimensional Hilbert space using an isometry $V:\mathbb{C}^r\rightarrow\mathbb{C}^d$. Then, any fidelity estimation algorithm on $V\rho V^\dag$ and $\frac{1}{r}V I_r V^\dag$ ($\frac{1}{r}I_r$ is the $r$-dimensional maximally mixed state) can be used to test whether $\rho$ is the maximally mixed state or $\sqrt{\varepsilon}$-far from the maximally mixed state in trace distance (due to the inequality $\mathrm{T}(\rho,\sigma)\leq \sqrt{1-\mathrm{F}(\rho,\sigma)^2}$~\cite{nielsen2010quantum}).
Therefore, the sample lower bound of $\Omega(\frac{r}{\varepsilon^2})$ in \cite[Theorem 1.10]{OW21} gives a lower bound of $\Omega(\frac{r}{\varepsilon})$ for the fidelity estimation task. 

Moreover, \cite[Lemma 13]{ALL22} shows that estimating $\tr(\rho\sigma)$ for pure states $\rho$ and $\sigma$ to within error $\varepsilon$ requires sample complexity $\Omega(\frac{1}{\varepsilon^2})$. Since $\tr(\rho\sigma)=\mathrm{F}(\rho,\sigma)^2$ when $\rho,\sigma$ are pure states, this lower bound also applies to the fidelity estimation task.

Therefore, we obtain a sample lower bound $\Omega(\frac{r}{\varepsilon}+\frac{1}{\varepsilon^2})$ and by \cref{thm:main} we obtain a query lower bound $\Omega(\frac{\sqrt{r}}{\sqrt{\varepsilon}}+\frac{1}{\varepsilon})$.
\end{proof}

When the two quantum states are well-conditioned, i.e., the reciprocal of the minimum non-zero eigenvalue of $\rho$ and $\sigma$ is (upper bounded by) $\kappa$, a quantum algorithm for fidelity estimation was proposed in \cite[Theorem 16]{LWWZ25} with query complexity $\widetilde{O}\rbra{\frac{\kappa^4}{\varepsilon}}$, which was later improved to $O\rbra{\frac{\kappa}{\varepsilon}\log\rbra{\frac{1}{\varepsilon}}}$ in \cite[Theorem 23]{UNWT25}. 
A nearly matching lower bound of $\Omega\rbra{\frac{1}{\varepsilon}}$ was given in \cite[Lemma 17]{LWWZ25} when $\kappa = \Theta\rbra{1}$. 

By quantum sample-to-query lifting, we can simply reproduce this result.

\begin{theorem} [{\cite[Lemma 17]{LWWZ25}} reproduced] \label{thm:lb-fidelity-kappa}
    The quantum query complexity of well-conditioned fidelity estimation is $\Omega(\frac{1}{\varepsilon})$ even if $\kappa = \Theta\rbra{1}$.
\end{theorem}
\begin{proof}
    To apply \cref{thm:main}, we only have to note that the sample lower bound for well-conditioned fidelity estimation is $\Omega\rbra{\frac{1}{\varepsilon^2}}$ even if $\kappa = \Theta\rbra{1}$, which was given in \cite[Lemma 40]{LWWZ25}. 
\end{proof}

\subsection{Other quantum tasks}

There are a series of lower bounds that can be derived by reducing from quantum state discrimination and applying \Cref{thm:main}.

Let $\textsc{Dis}_{\rho, \sigma}$ be a promise problem for quantum state discrimination; i.e., let $\textsc{Dis}_{\rho, \sigma}^{\textup{yes}}=\braces*{\rho}$ and $\textsc{Dis}_{\rho, \sigma}^{\textup{no}}=\braces*{\sigma}$.
We have the following lower bound for quantum state discrimination, which can be derived from the Helstrom-Holevo theorem~\cite{Helstrom69,Holevo73}.

\begin{lemma}[Lower bound for quantum state discrimination, cf.\ {\cite[Lemma 2.3]{WZ25b}}]
    \label{lmm:lower-discrimination}
    The quantum sample complexity for quantum state discrimination is lower bounded by $\mathsf{S}\rbra{\textsc{Dis}_{\rho, \sigma}} = \Omega\rbra{1/\gamma}$, where $\gamma = 1 - \operatorname{F}\rbra{\rho, \sigma}$ is the infidelity.
\end{lemma}

\subsubsection{Amplitude estimation}

Amplitude estimation is a basic quantum subroutine to estimate the amplitude $a$ of a quantum state $a \ket{0}\ket{\phi_0} + \sqrt{1-a^2} \ket{1} \ket{\phi_1}$, given its state preparation unitary $U$. 

\begin{problem}[Amplitude estimation]
    \label{prb:amp_est}
    Given queries to a unitary $U$ such that $U \ket{0} = a \ket{0}\ket{\phi_0} + \sqrt{1-a^2} \ket{1} \ket{\phi_1}$ with $a\in [0,1]$,
    the task is to estimate the amplitude $a$ to within additive error $\varepsilon$; i.e., output an $\hat a\in [0,1]$ such that
    $\abs*{\hat{a}-a}\leq \varepsilon$ with probability $\geq 2/3$.
\end{problem}

Using \Cref{thm:main} and the reduction in \cite{WZ25b}, we can reproduce the tight lower bound $\Omega(1/\varepsilon)$ for amplitude estimation.
Previously, this lower bound is obtained by reducing from quantum counting~\cite[Corollary 1.12]{NW99}, matching the upper bound $O\rbra{1/\varepsilon}$~\cite{BHMT02,Wan24}.
It is worth noting that in \Cref{prb:amp_est}, if one is only given query to the time-forward $U$ but not its time-reverse $U^\dagger$, the query lower bound will become $\Omega\parens*{\min\parens*{d,1/\varepsilon^2}}$~\cite{TW25}. 

\begin{theorem}[Lower bound for amplitude estimation~{\cite[Corollary 1.12]{NW99}} reproduced]
    \label{thm:lb_amp_est}
    The quantum query complexity of amplitude estimation is $\Omega(1/\varepsilon)$.
\end{theorem}

\begin{proof}
    We can reduce the quantum state discrimination problem $\textsc{Dis}_{\rho_+,\rho_-}$ to the amplitude estimation problem as in \cite{WZ25b}, by considering the following two states
    \[
    \rho_{\pm} = \rbra*{\frac 1 2 \mp \frac{5}{2} \varepsilon} \ket{0}\bra{0} + \rbra*{\frac 1 2 \pm \frac{5}{2} \varepsilon} \ket{1}\bra{1},
    \]
    with infidelity $O(\varepsilon^2)$.
    From \Cref{lmm:lower-discrimination}, we have a sample lower bound $\mathsf{S}\rbra{\textsc{Dis}_{\rho_+,\rho_-}} = \Omega(1/\varepsilon^2)$,
    which further implies a query lower bound $\mathsf{Q}\rbra{\textsc{Dis}_{\rho_+,\rho_-}} = \Omega(1/\varepsilon)$, using \Cref{thm:main}.
    Suppose $\calA^U$ is a quantum algorithm that solves the amplitude estimation problem to within error $\varepsilon$, using $Q$ queries to $U$.
    By taking $U=U_{\rho_{\pm}}$ such that
    \[
        U_{\rho_{\pm}}\ket{0}\ket{0}=\sqrt{\frac 1 2 \mp \frac{5}{2}\varepsilon}\ket{0}\ket{0} + \sqrt{\frac 1 2 \pm \frac{5}{2}\varepsilon}\ket{1}\ket{1},
    \]
    algorithm $\calA^U$ can solves $\textsc{Dis}_{\rho_+,\rho_-}$, 
    because $\sqrt{\frac 1 2+ \frac{5}{2}\varepsilon} - \sqrt{\frac 1 2 - \frac{5}{2}\varepsilon}> 3\varepsilon$.
    Since $U_{\rho_{\pm}}$ give purified query access to $\rho_{\pm}$,
    we obtain $Q\geq \mathsf{Q}\rbra{\textsc{Dis}_{\rho_+,\rho_-}} = \Omega(1/\varepsilon)$.
\end{proof}

\subsubsection{Hamiltonian simulation}

Hamiltonian simulation~\cite{Fey82} is a task to simulate the dynamics of the Schr{\"o}dinger equation governed by a Hamiltonian $H$. 
Specifically, we want to implement the unitary operator $e^{-iHt}$ for a given Hamiltonian $H$ and simulation time $t$.

\begin{problem}[Hamiltonian simulation]
    Given a block-encoding $U$ of a Hamiltonian $H$, the task is to approximately implement the evolution unitary $e^{-iHt}$ for an evolution time $t$ to a constant precision; i.e., to implement a quantum channel $\calE$ such that for any state $\rho$, the trace norm $\norm*{\calE(\rho)-e^{-iHt}\rho e^{iHt}}_1\leq 0.01$.
\end{problem}

Over the years, many efficient quantum algorithms for Hamiltonian simulation have been proposed, e.g., \cite{Llo96,ATS03,BACS07,CW12,BCC+14,BCC+15,BCK15,LC17,LC19,ZWY24}. The optimal query complexity with respect to the simulation time $t$ is $\Theta\rbra{t}$~\cite{LC17,LC19}.
The lower bound part $\Omega(t)$
is obtained by a reduction~\cite{BACS07,CK10,BCC+14,BCK15} from the parity problem, whose lower bound can be proved by the polynomial method \cite{BBC+01,FGGS98}.

Using \Cref{thm:main} and the reduction in \cite{WZ25b}, we can reproduce this tight lower bound $\Omega(t)$ for Hamiltonian simulation.

\begin{theorem}[Lower bound for Hamiltonian simulation, cf.\ \cite{BACS07,GSLW19}] \label{thm:hamiltonian}
    For $t\geq 2\pi$, the quantum query complexity for Hamiltonian simulation is $\Omega(t)$. 
\end{theorem}

\begin{proof}
    We can reduce from the quantum state discrimination problem $\textsc{Dis}_{\rho,\sigma}$ as in \cite{WZ25b}, by considering the following two states
    \[
        \rho = \frac 1 2 \ket{0}\bra{0} + \frac 1 2 \ket{1}\bra{1}\quad\textup{and}\quad\sigma = \rbra*{\frac 1 2 + \frac{\pi}{t}} \ket{0}\bra{0} + \rbra*{\frac 1 2 - \frac{\pi}{t}} \ket{1}\bra{1},
    \]
    with infideltiy $O(1/t^2)$. 
    Similar to the proof of \Cref{thm:lb_amp_est}, 
    by applying \Cref{lmm:lower-discrimination,thm:main}, we have query lower bound $\mathsf{Q}\rbra{\textsc{Dis}_{\rho,\sigma}} = \Omega(t)$. 
    Suppose $\calA^U$ is a quantum algorithm that simulates Hamiltonian for time $t$ problem to within error $0.01$, using $Q$ queries to $U$.
    Let $U_{\rho}$ and $U_{\sigma}$ be purified query access to $\rho$ and $\sigma$, respectively. Then, using \cite[Lemma 25]{GSLW19}, one can obtain a query to a block-encoding of $\rho$ (resp.\ $\sigma$) using $2$ queries to $U_{\rho}$ (resp.\ $U_{\sigma}$).
    Taking these block-encodings as $U$ in $\calA^U$ leads to implementations, using $Q$ queries, of $e^{-i\frac{\rho}{2} t}$ and $e^{-i\frac{\sigma}{2} t}$ to within error $0.01$. This can solve $\textsc{Dis}_{\rho,\sigma}$ by applying $\calA^U$ on state $\ket{+}$ followed by a Hadamard basis measurement:
    \[
        \abs*{\bra{+} e^{-i \frac{\rho}{2} t} \ket{+}}^2 - 0.01=0.99\geq \frac{2}{3}  \quad\textup{and}\quad \abs*{\bra{-}e^{-i \frac{\sigma}{2} t} \ket{+}}^2 +0.01=0.01\leq \frac{1}{3}.
    \]
    Thus, $Q\geq \mathsf{Q}\rbra{\textsc{Dis}_{\rho,\sigma}}=\Omega(t)$.
\end{proof}

\subsubsection{Quantum Gibbs sampling}

Quantum Gibbs sampling is a task to prepare the Gibbs (thermal) state $e^{-\beta H}/\mathcal{Z}$ of a given Hamiltonian $H$ at inverse temperature $\beta \geq 0$.
Here, $\mathcal{Z} = \tr\rbra{e^{-\beta H}}$ is the normalization factor,
also known as the partition function.

\begin{problem}[Quantum Gibbs sampling]
    \label{prb:gibbs}
    Given a block-encoding $U$ of a Hamiltonian $H$, the task is to prepare the Gibbs state $e^{-\beta H}/\calZ$ to within constant error at inverse temperature $\beta$, where $\calZ= \tr(e^{-\beta H})$ is the normalization factor;
    i.e., to prepare a quantum state $\rho$ such that
    the trace norm $\norm*{\rho-e^{-\beta H}/\calZ}_1\leq 0.01$.
\end{problem}

Many quantum algorithms for Gibbs sampling have been proposed~\cite{PW09,TOV+11,YAG12,KB16,CS17,BK19,CKG23,ZBC23,RFA24,DLL25,RS25,CKBG25,HSDS25}. 
Gibbs sampling is found to be useful in solving semidefinite programs \cite{BS17,vAGGdW20,BKL+19,vAG19}. 
Regarding \Cref{prb:gibbs}, a quantum algorithm with query complexity $O\rbra{\beta}$ was presented in~\cite{GSLW19}. 

Using \Cref{thm:main} and the reduction in \cite{WZ25b}, we can reproduce the tight lower bound $\Omega(\beta)$ for quantum Gibbs sampling.
Previously, this lower bound was obtained by reducing from Hamiltonian discrimination~\cite[Proposition G.5]{CKBG23} or unitary discrimination~\cite{Weg25}.

\begin{theorem}[Lower bound for Gibbs sampling]
    \label{thm:gibbs}
    For $\beta\geq 4$, the quantum query complexity for Gibbs sampling is $\Omega(\beta)$.
\end{theorem}

\begin{proof}
     We can reduce from the quantum state discrimination problem $\textsc{Dis}_{\rho_+,\rho_-}$ as in \cite{WZ25b}, by considering the following two states
     \[
    \rho_{\pm} = \rbra*{\frac 1 2 \mp \frac{2}{\beta}} \ket{0}\bra{0} + \rbra*{\frac 1 2 \pm \frac{2}{\beta}} \ket{1}\bra{1},
    \]
    with infidelity $O(1/\beta^2)$.
    Similar to the proof of \Cref{thm:lb_amp_est}, 
    by applying \Cref{lmm:lower-discrimination,thm:main}, we have query lower bound $\mathsf{Q}\rbra{\textsc{Dis}_{\rho_+,\rho_-}} = \Omega(\beta)$. 
    Suppose $\calA^U$ is a quantum algorithm that prepares the Gibbs state $e^{-\beta H/\calZ}$ at inverse temperature $\beta$ to within trace norm error $0.01$, using $Q$ queries to $U$. 
    Similar to the proof of \Cref{thm:hamiltonian}, 
    we can use $2$ queries to $U_\rho$ (resp.\ $U_\sigma$)
    to obtain a query to a block-encoding of $\rho$ (resp.\ $\sigma$), where $U_\rho$ (resp.\ $U_\sigma$) provides purified query access to $\rho$ (resp.\ $\sigma$).
    Taking these block-encodings as $U$ in $\calA^U$ leads to preparations, using $Q$ queries, of mixed states $\propto e^{-\beta\frac{\rho_{\pm}}{2}}$ to within error $0.01$.
    This can solve $\textsc{Dis}_{\rho_+,\rho_-}$ by applying $\calA^U$ followed by a computational basis measurement:
    \begin{align*}
        \tr\rbra*{\ket{0}\!\bra{0}\cdot \frac{e^{-\beta\frac{\rho_+}{2}}}{\tr\parens*{e^{-\beta\frac{\rho_+}{2}}}}}-0.01&=\frac{e}{e+e^{-1}}-0.01 \geq \frac{2}{3}\\
        \tr\rbra*{\ket{0}\!\bra{0}\cdot \frac{e^{-\beta\frac{\rho_-}{2}}}{\tr\parens*{e^{-\beta\frac{\rho_-}{2}}}}}+0.01&=\frac{e^{-1}}{e+e^{-1}}+0.01\leq  \frac{1}{3}.
    \end{align*}
    Thus, $Q\geq \mathsf{Q}\rbra{\textsc{Dis}_{\rho_+,\rho_-}} = \Omega(\beta)$.
\end{proof}

A variant of the standard Gibbs sampling in \Cref{prb:gibbs} is to replace the given oracle by a block-encoding of $\sqrt{H}$, formalized below.
In this setting, the query upper bound can achieve $O\rbra{\sqrt{\beta}}$ in some special cases~\cite{CS17}, e.g., when $H$ is a linear combination of Pauli operators. 

\begin{problem}[Gibbs sampling with stronger oracles]
    Given a block-encoding $U$ of the square root $\sqrt{H}$ of a Hamiltonian $H$, the task is to prepare the Gibbs state $e^{-\beta H}/\calZ$ to within constant error at inverse temperature $\beta$, as in \Cref{prb:gibbs}.
\end{problem}

Using \Cref{thm:main} and the reduction in \cite{WZ25b}, we can show that the lower bound $\Omega(\beta)$ for quantum Gibbs sampling still hold even with stronger oracles. This new lower bound improves the prior $\widetilde{\Omega}(\beta)$~\cite{WZ25b}.

\begin{theorem}[Impossibility of improving Gibbs sampling even with stronger oracles] \label{thm:stronger-gibbs}
    For $\beta\geq 4$, even with the stronger oracles in \Cref{prb:gibbs}, the quantum query complexity for Gibbs sampling is $\Omega(\beta)$.
\end{theorem}

\begin{proof}
    The proof is similar to that of \Cref{thm:gibbs} using the same construction.
    We only need to notice that
    \begin{align*}
        \tr\rbra*{\ket{0}\!\bra{0}\cdot \frac{e^{-\beta\frac{\rho_+}{2}}}{\tr\parens*{e^{-\beta\frac{\rho_+}{2}}}}}-0.01 &=\frac{e^{0.5}}{e^{0.5}+e^{-0.5}}-0.01\geq \frac{2}{3}\\
        \tr\rbra*{\ket{0}\!\bra{0}\cdot \frac{e^{-\beta\frac{\rho_-}{2}}}{\tr\parens*{e^{-\beta\frac{\rho_-}{2}}}}}+0.01 &=\frac{e^{-0.5}}{e^{0.5}+e^{-0.5}}+0.01\leq  \frac{1}{3}.
    \end{align*}
\end{proof}

\subsubsection{The entanglement entropy problem}

The entanglement entropy problem~\cite{SY23} is a problem of unitary property testing that decides whether a bipartite state $\ket{\psi} \in \calH_{\mathsf{A}} \otimes \calH_{\mathsf{B}}$ has low ($\leq a$) or high ($\geq b$) R\'enyi entanglement entropy.

\begin{problem} [The entanglement entropy problem {\cite[Definition 1.13]{SY23}}]
    \label{prb:ent_entropy}
    Given query access to the reflection operator $U_{\psi}=I-2\ketbra{\psi}{\psi}$ about a bipartite pure state $\ket{\psi} \in \mathcal{H}_{\mathsf{A}} \otimes \mathcal{H}_{\mathsf{B}}$ with $\mathcal{H}_{\mathsf{A}} = \mathcal{H}_{\mathsf{B}} = \mathbb{C}^d$, the entanglement entropy problem $\calP=\textsc{QEntanglementEntropy}_{\alpha, a, b}$ with $0 < a < b \leq \ln\rbra{d}$ is a promise problem for testing whether a bipartite pure state has low ($\leq a$) or high ($\geq b$) entanglement entropy, with
    \begin{align*}
        \calP^{\textup{yes}}&=\braces*{\ket{\psi}\in \calH_{\mathsf{A}}\otimes \calH_{\mathsf{B}}:\mathrm{S}_\alpha^{\textup{R\'en}}\rbra*{\tr_{\mathsf{B}}\rbra{\ketbra{\psi}{\psi}}} \leq a},\\ \calP^{\textup{no}}&=\braces*{\ket{\psi}\in \calH_{\mathsf{A}}\otimes \calH_{\mathsf{B}}:\mathrm{S}_\alpha^{\textup{R\'en}}\rbra*{\tr_{\mathsf{B}}\rbra{\ketbra{\psi}{\psi}}} \geq b},
    \end{align*}
    where $\mathrm{S}_\alpha^{\textup{R\'en}}\rbra{\cdot}$ is the $\alpha$-R\'enyi entropy.
\end{problem}

Previously, a query lower bound of $\Omega\rbra{e^{a/4}}$ was given in \cite[Theorem 1.14]{SY23} (when $\alpha = 2$), which gives $\Omega\rbra{d^{1/4}}$ when $a$ is close to $\ln\rbra{d}$. 
Later improvements focused on the dependence on $\Delta = b - a$.
In \cite[Theorem 4.3]{WZ25b}, a query lower bound of $\widetilde{\Omega}\rbra{\frac{1}{\sqrt{\Delta}}}$ was given, which was improved to $\Omega\rbra{\frac{1}{\sqrt{\Delta}}}$ in \cite[Claim 2]{Weg25}. 
Later, a query lower bound of $\widetilde{\Omega}\rbra{\frac{\sqrt{d}}{\sqrt{\Delta}}+\frac{d^{\frac{1}{\alpha}-1}}{\Delta^{\frac{1}{\alpha}}}}$ in terms of both $d$ and $\Delta$ was shown in \cite[Theorem 5.1]{CWZ24} for every $\alpha > 0$. 

Using \Cref{thm:main}, we are able to prove new quantum query lower bounds for the entanglement entropy problem. 

\begin{theorem} \label{thm:ent-entropy-prob}
    Let $\calP_{\alpha, a, b}=\textsc{QEntanglementEntropy}_{\alpha, a, b}$ be the entanglement entropy problem.
    The quantum query complexity for $\calP$ has the following lower bounds, where $\Delta=b-a$:
    \begin{itemize}
        
        \item 
            For $\alpha =1$: $\mathsf{Q}_\diamond^{\mathsf{refl}}(\calP_{\alpha, \ln\rbra{d}-\Delta, \ln\rbra{d}})=\Omega\rbra{\frac{\sqrt{d}}{\sqrt{\Delta}}}$.
            Moreover, there exists an $a \in \sbra{0, \ln\rbra{d}-\Delta}$ such that $\mathsf{Q}_\diamond^{\mathsf{refl}}(\calP_{\alpha, a, a+\Delta}) = \Omega\rbra{\frac{\frac{\sqrt{d}}{\sqrt{\Delta}}+\frac{\log\rbra{d}}{\Delta}}{\sqrt{\log\rbra{\frac{1}{\Delta}}\log\log\rbra{\frac{1}{\Delta}}}}}$.
        \item 
            For $0 < \alpha < 1$:
            $\mathsf{Q}_\diamond^{\mathsf{refl}}(\calP_{\alpha, \ln\rbra{d}-\Delta, \ln\rbra{d}}) = \Omega\rbra{\frac{\sqrt{d}}{\sqrt{\Delta}}}$ and          $\mathsf{Q}_\diamond^{\mathsf{refl}}(\calP_{\alpha, 0, \Delta}) = \Omega\rbra{\frac{d^{\frac{1}{2\alpha}-\frac{1}{2}}}{\Delta^{\frac{1}{2\alpha}}}}$. 
            Moreover, there exists a constant $\Delta > 0$ and $a \in \sbra{0, \ln\rbra{d}-\Delta}$ such that $\mathsf{Q}_\diamond^{\mathsf{refl}}(\calP_{\alpha, a, a+\Delta}) = \Omega\rbra{d^{\frac{1}{2\alpha}-o\rbra{1}}}$. 
            
        \item 
            For $\alpha > 1$:
            \begin{itemize}
                \item 
                    If $\alpha$ is integer, then $\mathsf{Q}_\diamond^{\mathsf{refl}}(\calP_{\alpha,\ln\rbra{d}-\Delta,\ln\rbra{d}}) = \Omega\rbra{\frac{d^{1-\frac{1}{\alpha}}}{\Delta^{\frac{1}{\alpha}}}}$ for $\Delta \geq \Omega\rbra{\frac{1}{\sqrt{d}}}$. Moreover, there exists an $a \in \sbra{0, \ln\rbra{d}-\Delta}$ such that $\mathsf{Q}_\diamond^{\mathsf{refl}}(\calP_{\alpha,a,a+\Delta}) = \Omega\rbra{\frac{\frac{d^{1-\frac{1}{\alpha}}}{\Delta^{\frac{1}{\alpha}}}+\frac{d^{\frac{1}{2}-\frac{1}{2\alpha}}}{\Delta}}{\sqrt{\log\rbra{\frac{1}{\Delta}}\log\log\rbra{\frac{1}{\Delta}}}}}$. 
                \item 
                    If $\alpha$ is non-integer, $\mathsf{Q}_\diamond^{\mathsf{refl}}(\calP_{\alpha, \ln\rbra{d}-\Delta, \ln\rbra{d}})=\Omega\rbra{\frac{\sqrt{d}}{\sqrt{\Delta}}}$. 
                    Moreover, there exists an $a \in \sbra{0, \ln\rbra{d}-\Delta}$ such that $\mathsf{Q}_\diamond^{\mathsf{refl}}(\calP_{\alpha,a,a+\Delta}) = \Omega\rbra{\frac{\frac{\sqrt{d}}{\sqrt{\Delta}}+\frac{d^{\frac{1}{2}-\frac{1}{2\alpha}}}{\Delta}}{\sqrt{\log\rbra{\frac{1}{\Delta}}\log\log\rbra{\frac{1}{\Delta}}}}}$. 
            \end{itemize}
    \end{itemize}
\end{theorem}

\begin{proof}
    It is shown in~\cite{CWZ24} that $\textsc{QEntanglementEntropy}_{\alpha, a, b}$ can be reduced from $\textsc{Entropy}_{\alpha, a, b}$, a promise problem that tests whether the $\alpha$-R\'enyi entropy of a given mixed state is low ($\leq a$) or high ($\geq b$), defined similarly to \Cref{prb:ent_entropy}.
    Equivalently, we have $\mathsf{S}(\textsc{QEntanglementEntropy}_{\alpha, a, b})=\mathsf{S}(\textsc{Entropy}_{\alpha, a, b})$. 

    Let us consider the sample lower bound for $\mathsf{S}(\textsc{Entropy}_{\alpha, a, b})$ case by case:
    \begin{itemize}
        \item
            For $\alpha =1$, we have $\mathsf{S}(\textsc{Entropy}_{\alpha, \ln\rbra{d}-\Delta, \ln\rbra{d}})=\Omega\rbra{\frac{d}{\Delta}}$ from the hard instance constructed in~\cite[Theorem VI.2]{WZ25}.
             For the second part, we can use a similar proof to that of \cref{thm:von-neumann-entropy} to obtain a sample lower bound for the estimating the von Neumann entropy.
             To achieve a sample lower bound for $\textsc{Entropy}_{\alpha, a, a+\Delta}$, we can use the following reduction from entropy estimation to entropy testing via a ternary search on the entropy value:
             each step we split the interval of possible entropy values into three intervals of equal lengths, and applying any algorithm for entropy testing enables us to shrink the interval length by $1/3$.
             We only need to take $O\rbra{\log\rbra{\frac{1}{\Delta}}}$ steps to obtain an algorithm for entropy estimation, where in each step we repeat the testing for $O\rbra{\log\log\rbra{\frac{1}{\Delta}}}$ times to boost the final success probability to be $\geq 2/3$.
             Consequently, the von Neumann entropy can be estimated to within additive error $\Delta$ with sample complexity $\sup_{a \in \sbra{0, \ln\rbra{d}-\Delta}} \mathsf{S}(\textsc{Entropy}_{\alpha, a, a+\Delta}) \cdot O\rbra{\log\rbra{\frac{1}{\Delta}}\log\log\rbra{\frac{1}{\Delta}}}$. 
             Given the sample lower bound $\Omega\rbra{\frac{d}{\varepsilon}+\frac{\log^2\rbra{d}}{\varepsilon^2}}$ for von Neumann entropy estimation mentioned in \cref{thm:von-neumann-entropy}, we have that there exists an $a \in \sbra{0, \ln\rbra{d}-\Delta}$ such that $\mathsf{S}(\textsc{Entropy}_{\alpha, a, a+\Delta}) = \Omega\rbra{\frac{\frac{d}{\Delta}+\frac{\log^2\rbra{d}}{\Delta^2}}{\log\rbra{\frac{1}{\Delta}}\log\log\rbra{\frac{1}{\Delta}}}}$.
        \item 
            For $0 < \alpha < 1$, similarly, 
            we have $\mathsf{S}(\textsc{Entropy}_{\alpha, \ln\rbra{d}-\Delta, \ln\rbra{d}}) = \Omega\rbra{\frac{d}{\Delta}}$  and $\mathsf{S}(\textsc{Entropy}_{\alpha, 0, \Delta}) = \Omega\rbra{\frac{d^{\frac{1}{\alpha}-1}}{\Delta^{\frac{1}{\alpha}}}}$ from~\cite[Theorem VI.5]{WZ25}.
            For the second part, using a similar proof to that of \cref{thm:lt1-renyi} and the previous reduction for the case of $\alpha=1$,
            there exists a constant $\Delta > 0$ and $a \in \sbra{0, \ln\rbra{d}-\Delta}$ such that $\mathsf{S}(\textsc{Entropy}_{\alpha, a, a+\Delta}) = \Omega\rbra{d^{\frac{1}{\alpha}-o\rbra{1}}}$.
        \item
        For $\alpha > 1$:
            \begin{itemize}
                \item 
                    If $\alpha$ is an integer, then $\mathsf{S}(\textsc{Entropy}_{\alpha,\ln\rbra{d}-\Delta,\ln\rbra{d}}) = \Omega\rbra{\frac{d^{2-\frac{2}{\alpha}}}{\Delta^{\frac{2}{\alpha}}}}$ for $\Delta \geq \Omega\rbra{\frac{1}{\sqrt{d}}}$ from the hard instance in~\cite[Lemma 9]{AISW20}.
                    For the second part, using a similar proof to that of \cref{thm:qRenyi} and the previous reduction for the case of $\alpha=1$,
                    there exists an $a \in \sbra{0, \ln\rbra{d}-\Delta}$ such that $\mathsf{S}(\textsc{Entropy}_{\alpha,a,a+\Delta}) = \Omega\rbra{\frac{\frac{d^{2-\frac{2}{\alpha}}}{\Delta^{\frac{2}{\alpha}}}+\frac{d^{1-\frac{1}{\alpha}}}{\Delta^2}}{\log\rbra{\frac{1}{\Delta}}\log\log\rbra{\frac{1}{\Delta}}}}$. 
                \item 
                    Similarly, if $\alpha$ is a non-integer, then $\mathsf{S}(\textsc{Entropy}_{\alpha, \ln\rbra{d}-\Delta, \ln\rbra{d}})=\Omega\rbra{\frac{d}{\Delta}}$ from the hard instance in~\cite[Theorem VI.4]{WZ25}.
                    For the second part, using a similar proof to that of \cref{thm:qRenyi} and the previous reduction for the case of $\alpha=1$,
                    there exists an $a \in \sbra{0, \ln\rbra{d}-\Delta}$ such that $\mathsf{S}(\textsc{Entropy}_{\alpha,a,a+\Delta}) = \Omega\rbra{\frac{\frac{d}{\Delta}+\frac{d^{1-\frac{1}{\alpha}}}{\Delta^2}}{\log\rbra{\frac{1}{\Delta}}\log\log\rbra{\frac{1}{\Delta}}}}$.
            \end{itemize}
    \end{itemize}
    Finally, using \Cref{corollary:lifting-reflection}, the conclusion immediately follows.
\end{proof}

A quantum algorithm was presented in \cite[Proposition 1]{Weg25} with query complexity $O\rbra{\frac{d}{\Delta}}$ for the case of $\alpha = 2$. 
For the same case, \cref{thm:ent-entropy-prob} gives a lower bound of $\Omega\rbra{\frac{1}{\Delta\sqrt{\log\rbra{\frac{1}{\Delta}}\log\log\rbra{\frac{1}{\Delta}}}}} = \widetilde{\Omega}\rbra{\frac{1}{\Delta}}$ in $\Delta$, which shows that the quantum algorithm in \cite{Weg25} for the entanglement entropy problem is near-optimal in $\Delta$. 

\section{List of Quantum Sample Upper Bounds} \label{sec:upper}

\subsection{Entropy estimation}

\subsubsection{Von Neumann entropy}

The current best quantum sample upper bound for von Neumann entropy estimation is $O\rbra{\frac{d^2}{\varepsilon^2}}$ due to \cite{AISW20} and $O\rbra{\frac{d^2}{\varepsilon}+\frac{\log^2\rbra{d}}{\varepsilon^2}}$ due to \cite{BMW16}.

By the quantum samplizer, we can obtain a quantum sample upper bound for von Neumann entropy estimation in the low-rank case.

\begin{theorem} \label{thm:von-upper}
    The sample complexity of von Neumann entropy estimation is $\widetilde{O}\rbra{\frac{r^2}{\varepsilon^4}}$ when the rank of the quantum state is $r$. 
\end{theorem}
\begin{proof}
    To apply \cref{thm:samplizer}, we only have to note the quantum query upper bound $\widetilde{O}\rbra{\frac{r}{\varepsilon^2}}$ for von Neumann entropy estimation given in \cite{WGL+24}. 
\end{proof}

\subsubsection{R\'enyi entropy}

The current best quantum sample upper bound for $\alpha$-R\'enyi entropy estimation is $O\rbra{\frac{d^{\frac{2}{\alpha}}}{\varepsilon^{\frac{2}{\alpha}}}}$ for $0 < \alpha < 1$ and $O\rbra{\frac{d^2}{\varepsilon^2}}$ for non-integer $\alpha > 1$ due to \cite{AISW20}. 

By the quantum samplizer, we can obtain a quantum sample upper bound for $\alpha$-R\'enyi entropy estimation in the low-rank case.

\begin{theorem} \label{thm:renyi-upper}
    The sample complexity of $\alpha$-R\'enyi entropy estimation is $\widetilde{O}\rbra{\frac{r^{\frac{2}{\alpha}}}{\varepsilon^{\frac{2}{\alpha}+2}}}$ for $0 < \alpha < 1$ and $\widetilde{O}\rbra{\frac{r^2}{\varepsilon^{2+\frac{2}{\alpha}}}}$ for $\alpha > 1$, when the rank of the quantum state is $r$. 
\end{theorem}
\begin{proof}
    To apply \cref{thm:samplizer}, we only have to note the quantum query upper bound $\widetilde{O}\rbra{\frac{r^{\frac{1}{\alpha}}}{\varepsilon^{\frac{1}{\alpha}+1}}}$ for $0 < \alpha < 1$ and $\widetilde{O}\rbra{\frac{r}{\varepsilon^{1+\frac{1}{\alpha}}}}$ for non-integer $\alpha > 1$ for $\alpha$-R\'enyi entropy estimation given in \cite{WZL24}. 
\end{proof}

\subsection{Closeness estimation}

\subsubsection{Trace distance}

The current best quantum sample upper bound for trace distance estimation is $O\rbra{\frac{r^2}{\varepsilon^5}\log^2\rbra{\frac{r}{\varepsilon}}\log^2\rbra{\frac{1}{\varepsilon}}}$ due to \cite{WZ24}, where $r$ is the rank of quantum states. 
For the pure-state case where $r = 1$, an optimal quantum algorithm for trace distance estimation was proposed in \cite{WZ24c} with sample complexity $\Theta\rbra{\frac{1}{\varepsilon^2}}$. 

By the quantum samplizer, we can obtain an improved quantum sample upper bound.

\begin{theorem} \label{thm:qs-td}
    The sample complexity of trace distance estimation is $O\rbra{\frac{r^2}{\varepsilon^4}\log^2\rbra{\frac{1}{\varepsilon}}}$, where $r$ is the rank of quantum states. 
\end{theorem}
\begin{proof}
    To apply \cref{thm:samplizer}, we only have to note the quantum query upper bound $O\rbra{\frac{r}{\varepsilon^2}\log\rbra{\frac{1}{\varepsilon}}}$ for trace distance estimation given in \cite{WZ24}. 
\end{proof}

\subsubsection{Fidelity}

The quantum sample upper bound for fidelity estimation is known to be $\widetilde{O}\rbra{\frac{r^{5.5}}{\varepsilon^{12}}}$ due to \cite{GP22} and $\widetilde{O}\rbra{\frac{\kappa^9}{\varepsilon^3}}$ due to \cite{LWWZ25}, where $r$ is the rank of quantum states and $\kappa$ is the ``condition number'' such that $\rho, \sigma \geq \frac{I}{\kappa}$. 
For the pure-state case where $r = 1$, an optimal quantum algorithm for fidelity estimation was proposed in \cite{WZ24c} with sample complexity $\Theta\rbra{\frac{1}{\varepsilon^2}}$. 

By the quantum samplizer, we can reproduce the quantum sample upper bounds in \cite[Theorem 25]{UNWT25}.

\begin{theorem}[{\cite[Theorem 25]{UNWT25}} reproduced] \label{thm:qs-fidelity-rank}
    The sample complexity of fidelity estimation is $O\rbra{\frac{r^2}{\varepsilon^4}\log^2\rbra{\frac{1}{\varepsilon}}}$, where $r$ is the rank of quantum states. 
\end{theorem}

\begin{proof}
    To apply \cref{thm:samplizer}, we only have to note the quantum query upper bound $O\rbra{\frac{r}{\varepsilon^2}\log\rbra{\frac{1}{\varepsilon}}}$ for fidelity estimation given in \cite[Theorem 23]{UNWT25}. 
\end{proof}

\begin{theorem}[{\cite[Theorem 25]{UNWT25}} reproduced] \label{thm:kappa-fidelity}
    The sample complexity of well-conditioned fidelity estimation is $O\rbra{\frac{\kappa^2}{\varepsilon^2}\log^2\rbra{\frac{1}{\varepsilon}}}$. 
\end{theorem}

\begin{proof}
    To apply \cref{thm:samplizer}, we only have to note the quantum query upper bound $O\rbra{\frac{\kappa}{\varepsilon}\log\rbra{\frac{1}{\varepsilon}}}$ for fidelity estimation given in \cite[Theorem 23]{UNWT25}. 
\end{proof}

\cref{thm:kappa-fidelity} achieves an optimal dependence on $\varepsilon$ up to polylogarithmic factors, which matches the lower bound $\Omega\rbra{\frac{1}{\varepsilon^2}}$ given in \cite[Lemma 40]{LWWZ25} when $\kappa = \Theta\rbra{1}$. 

\subsubsection{Quantum \texorpdfstring{$\ell_\alpha$}{ℓ\_α} distance}

\begin{problem}[Quantum $\ell_\alpha$ distance estimation]
    Given two unknown quantum states $\rho$ and $\sigma$, estimate the $\ell_\alpha$ distance
    \[
    \mathrm{T}_{\alpha}\rbra{\rho, \sigma} = \frac{1}{2} \rbra*{ \tr\rbra{\abs{\rho - \sigma}^{\alpha}} }^{1/\alpha}
    \]
    to within additive error $\varepsilon$. 
\end{problem}

It was noted in \cite{WGL+24} that quantum $\ell_\alpha$ distance can be estimated with sample complexity $\poly\rbra{r, \frac{1}{\varepsilon}}$, where the two unknown quantum states have rank $r$. 
The current best quantum sample upper bound for quantum $\ell_\alpha$ distance estimation is $\widetilde{O}\rbra{\frac{1}{\varepsilon^{3\alpha+2+\frac{2}{\alpha-1}}}}$ for constant $\alpha > 1$ due to \cite{LW25b}, removing the previous $r$-dependence. 
Moreover, a lower bound of $\Omega\rbra{\frac{1}{\varepsilon^2}}$ was also given in \cite{LW25b}.

By the samplizer, we can obtain an improved quantum sample upper bound. 

\begin{theorem} \label{thm:lalpha}
    For $\alpha > 1$, the sample complexity of quantum $\ell_\alpha$ distance estimation is $O\rbra{\frac{1}{\varepsilon^{2\alpha+2+\frac{2}{\alpha-1}}}}$. 
\end{theorem}
\begin{proof}
    To apply \cref{thm:samplizer}, we only have to note the quantum query upper bound $O\rbra{\frac{1}{\varepsilon^{\alpha+1+\frac{1}{\alpha-1}}}}$ for quantum $\ell_\alpha$ distance estimation given in \cite[Theorem 14]{LW25b}. 
\end{proof}

\subsubsection{Tsallis relative entropy}

The $\alpha$-Tsallis relative entropy of a quantum state $\rho$ with respect to another quantum state $\sigma$ is defined by
\[
\mathrm{D}_{\alpha}^{\textup{Tsa}}\rbra{\rho\|\sigma} = \frac{1}{1-\alpha} \rbra*{ 1 - \tr\rbra{\rho^\alpha \sigma^{1-\alpha}} }, \quad 0 < \alpha < 1.
\]
In particular, the case of $\alpha = \frac{1}{2}$ refers to the squared Hellinger distance \cite{LZ04}, defined by
\[
d_{\textup{H}}^2\rbra{\rho, \sigma} = \frac{1}{2} \mathrm{D}_{\frac{1}{2}}^{\textup{Tsa}}\rbra{\rho\|\sigma} =  1 - \tr\rbra*{\sqrt{\rho}\sqrt{\sigma}}. 
\]

\begin{problem}[Quantum Tsallis relative entropy estimation]
    Given two unknown quantum states $\rho$ and $\sigma$ of rank $r$, estimate the $\mathrm{D}_{\alpha}^{\textup{Tsa}}\rbra{\rho\|\sigma}$ to within additive error $\varepsilon$. 
\end{problem}

The current best quantum sample upper bound for $\alpha$-Tsallis relative entropy estimation is given by \cite{BGW25}, which is $O\rbra{\frac{r^{2+3\alpha}}{\varepsilon^{\frac{2}{\alpha} + \frac{3}{1-\alpha}}}\log^{2}\rbra{\frac{r}{\varepsilon}}}$ for $0 < \alpha < \frac{1}{2}$, $O\rbra{\frac{r^{5-3\alpha}}{\varepsilon^{\frac{2}{1-\alpha}+ \frac{3}{\alpha}}}\log^{2}\rbra{\frac{r}{\varepsilon}}}$ for $\frac{1}{2} < \alpha < 1$, and $O\rbra{\frac{r^{3.5}}{\varepsilon^{10}}\log^4\rbra{\frac{r}{\varepsilon}}}$ for $\alpha = \frac{1}{2}$. 

By the quantum samplizer, we can obtain a quantum sample upper bound for quantum $\alpha$-Tsallis relative entropy estimation.

\begin{theorem} \label{thm:tsallis-relative-upper}
    The sample complexity of quantum $\alpha$-Tsallis relative entropy estimation is $O\rbra{\frac{r^{2+2\alpha}}{\varepsilon^{\frac{2}{\alpha} + \frac{2}{1-\alpha}}}}$ for $0 < \alpha < \frac{1}{2}$, $O\rbra{\frac{r^{4-2\alpha}}{\epsilon^{\frac{2}{1-\alpha}+ \frac{2}{\alpha}}}}$ for $\frac{1}{2} < \alpha < 1$, and $O\rbra{\frac{r^3}{\varepsilon^8}\log^2\rbra{\frac{r}{\varepsilon}}}$ for $\alpha = \frac{1}{2}$, when the rank of the quantum state is $r$. 
\end{theorem}
\begin{proof}
    To apply \cref{thm:samplizer}, we only have to note the quantum query upper bound given by \cite{BGW25}, which is $O\rbra{\frac{r^{1+\alpha}}{\varepsilon^{\frac{1}{\alpha} + \frac{1}{1-\alpha}}}}$ for $0 < \alpha < \frac{1}{2}$, $O\rbra{\frac{r^{2-\alpha}}{\epsilon^{\frac{1}{1-\alpha}+ \frac{1}{\alpha}}}}$ for $\frac{1}{2} < \alpha < 1$, and $O\rbra{\frac{r^{1.5}}{\varepsilon^4}\log\rbra{\frac{r}{\varepsilon}}}$ for $\alpha = \frac{1}{2}$.
\end{proof}

\section{Discussion} \label{sec:discussion}

In this paper, we compile a list of quantum complexity bounds for property testing problems, thus exhibiting the power of quantum sample-to-query lifting. 
Using this powerful tool, we think that it is also possible to improve the quantum singular value transformation for quantum channels recently proposed by \cite{NRTM25}. 
In addition, the techniques of \cite{TWZ25} were recently shown to be useful in quantum state tomography \cite{PSTW25}. 

To conclude this paper, we ask several open questions for future research. 

\begin{enumerate}
    \item Can we prove lower bounds for the small-error case by the lifting method?

    In this paper, all quantum query lower bounds assume that the error probability is a small constant, e.g., $\leq \frac{1}{3}$. 
    However, it is not clear how to use the lifting method to derive quantum query lower bounds when the error probability is extremely small. 
    By contrast, the quantum polynomial method \cite{BBC+01} has been shown to be useful for the small-error case, e.g., \cite{BCdWZ99,MdW23}. 

    \item Can we prove a matching quantum query lower bound for solving systems of linear equations by the lifting method?

    For a (sparse) matrix $A$ with $I/\kappa \leq A \leq I$ for some $\kappa > 0$, preparing the pure state $\ket{x} \propto A^{-1} \ket{0}$ is known to require quantum query complexity $\Omega\rbra{\kappa}$ \cite{HHL09,OD21,HK21,WZ24b}. 
    For example, the approach in \cite{OD21} uses the lower bound for computing partial Boolean functions in \cite{NW99}, which traces back to the quantum polynomial method \cite{BBC+01}. 
    In comparison, a proof using the lifting method will provide an information-theoretic understanding of solving systems of linear equations. 

    \item Can we show prove a matching sample lower bound of $\Omega\rbra{d^2}$ for von Neumann entropy estimation, trace distance estimation, and fidelity estimation?

    The sample complexity of the uniformity testing of probability distributions is known to be $\Theta\rbra{\sqrt{d}}$ \cite{Pan08,CDVV14}, whereas its closeness estimation counterparts were shown to have a near-quadratic blowup. 
    For example, the sample complexities of Shannon entropy estimation \cite{VV11a,JVHW15,WY16} and total variation distance estimation \cite{VV17} are known to be $\Theta\rbra{\frac{d}{\log\rbra{d}}}$.
    Here, for convenience, we assume the precision parameter $\varepsilon = \Theta\rbra{1}$ is a small enough constant. 
    In the quantum case, the sample complexity of the mixedness testing of quantum states (the quantum counterpart of the uniformity testing of probability distributions) is known to be $\Theta\rbra{d}$ \cite{OW21}. 
    As also suggested in \cite{Wri23}, we therefore conjecture that a similar quadratic blowup also appears in the quantum case. 

    \begin{conjecture} \label{conj:s-vN}
        The quantum sample complexities of von Neumann entropy estimation, trace distance estimation, and fidelity estimation are $\Omega\rbra{d^2}$. 
    \end{conjecture}

    \cref{conj:s-vN} means that the current best sample upper bounds for von Neumann entropy estimation \cite{AISW20,BMW16,WZ25}, trace distance estimation \cite{WZ24}, and fidelity estimation \cite{UNWT25} are optimal. 
    Moreover, if \cref{conj:s-vN} is true, then by quantum sample-to-query lifting, we can further prove matching query lower bounds for these problems. 

    \begin{proposition} \label{prop:conj}
        If \cref{conj:s-vN} is true, then the quantum query complexities of von Neumann entropy estimation, trace distance estimation, and fidelity estimation are $\Omega\rbra{d}$. 
    \end{proposition}

    Notably, \cref{prop:conj} means that if \cref{conj:s-vN} is true, then the current best query upper bounds for von Neumann entropy estimation \cite{GL20,WGL+24}, trace distance estimation \cite{WZ24}, and fidelity estimation \cite{UNWT25} are also optimal. 
\end{enumerate}

\addcontentsline{toc}{section}{References}

\bibliographystyle{alphaurl}
\bibliography{main}

\end{document}